\documentclass[12pt,draft]{article}

\usepackage{amsthm,amssymb,amsmath,amscd,amstext,comment,latexsym,mathrsfs}
\usepackage{mathrsfs,mathptmx,tikz,xypic}

\usepackage[OT2,OT1]{fontenc}
\newcommand\cyr
{\renewcommand\rmdefault{wncyr}
\renewcommand\sfdefault{wncyss}
\renewcommand\encodingdefault{OT2}
\normalfont
\selectfont}
\DeclareTextFontCommand{\textcyr}{\cyr}

\def\Eoborotnoye{\char"003}

\newtheorem{definition}{Definition}[section]
\newtheorem{theorem}{Theorem}[section]
\newtheorem{lemma}{Lemma}[section]

\newcommand{\nc}{\newcommand}

\nc{\C}{{\mathbb C}}
\nc{\R}{{\mathbb R}}
\nc{\HH}{{\mathbb H}}
\nc{\Z}{{\mathbb Z}}
\nc{\N}{{\mathbb N}}
\nc{\s}{{\mathbb S}}

\nc{\dd}{{\rm d}}

\nc{\ii}{{\bf i}}

\nc{\cg}{{\mathscr G}}
\nc{\ch}{{\mathscr H}}
\nc{\ci}{{\mathscr I}}
\nc{\cj}{{\mathscr J}}
\nc{\cv}{{\mathscr V}}

\nc{\fb}{{\mathfrak B}}
\nc{\ff}{{\mathfrak F}}
\nc{\fii}{{\mathfrak I}}
\nc{\fj}{{\mathfrak J}}
\nc{\fr}{{\mathfrak R}}
\nc{\fm}{{\mathfrak M}}

\nc{\sm}{{\mathsf M}{\mathsf a}{\mathsf n}^4}

\DeclareMathOperator{\re}{Re}
\DeclareMathOperator{\ed}{End}

\nc{\stack}[2]{{\begin{array}{c}
\scriptstyle #1 \\ \scriptstyle #2 \end{array}} }

\begin{document}

\title{An operator algebraic characterization of the Riemannian 
vacuum Einstein equation in four dimensions}

\author{G\'abor Etesi\\
\small{{\it Department of Algebra and Geometry, Institute of Mathematics,}}\\
\small{{\it Budapest University of Technology and Economics,}}\\
\small{{\it M\H uegyetem rkp. 3., H-1111 Budapest, Hungary}}
\footnote{E-mail: {\tt etesi@math.bme.hu, etesigabor@gmail.com}}}

\maketitle

\pagestyle{myheadings}
\markright{Operator algebraic characterization of the Einstein equation}

\thispagestyle{empty}

\begin{abstract} 
In this paper, using connected compact oriented smooth 
$4$-manifolds, some representations of the hyperfinite ${\rm II}_1$-type 
factor von Neumann algebra are constructed. The Murray--von Neumann 
coupling constant of these representations gives rise to a new smooth 
$4$-manifold invariant whose very first properties are investigated. 

Moreover as a part of this construction, a connected compact oriented smooth 
$4$-manifold admits an embedding into the hyperfinite ${\rm II}_1$ 
factor. This embedding, on the one hand, induces a Riemannian metric on 
the manifold such that its Riemannian curvature tensor belongs to the von 
Neumann algebra; on the other hand the 
metric induces a periodic dynamics on the von Neumann algebra, what we 
call the Hodge dynamics on the hyperfinite ${\rm II}_1$ factor. It is 
observed that the metric is Einstein {\it i.e.}, satisfies the (Riemannian) 
vacuum Einstein equation with a possibly non-zero cosmological constant, 
if and only if its Riemannian curvature tensor belongs to the 
fixed-point-subalgebra of the Hodge dynamics.

Finally, we make a comprehensive enumeration of all representations of 
the hyperfinite ${\rm II}_1$ factor constructed here, from the viewpoint 
of thermal equilibrium states and phase transitions in algebraic quantum 
field theory. 
\end{abstract}

\centerline{AMS Classification: Primary: 46L10, 83C45, Secondary: 57R55}
\centerline{Keywords: {\it Smooth $4$-manifold; Hyperfinite ${\rm II}_1$
factor; Einstein equation; Quantum gravity}}

%%%%%%%%%%%%%%%%%%%%%%%%%%%%%%%%%%%%%%%%%

\section{Introduction and summary}
\label{one}

%%%%%%%%%%%%%%%%%%%%%%%%%%%%%%%%%%%

The {\it hyperfinite factor von Neumann algebra of type ${\rm II}_1$} is 
distinguished among von Neumann algebras in many senses. Apparently this 
was von Neumann's favourite operator algebra and he was especially 
satisfied with its discovery. As it is known (cf. \cite[pp. 22-32]{neu} for 
a possible reconstruction of the story) he 
attempted, but finally did not complete or abandoned, to use 
the hyperfinite ${\rm II}_1$ factor to bring quantum 
mechanics to a not only mathematically but even conceptually sound 
basis, by interpreting quantum probabilities as relative frequencies of 
a particular statistical ensemble sorted from an absolute one, provided by 
the ${\rm II}_1$ hyperfinite factor 
with its unique normalized trace. The idea looked appealing not only by 
the uniqueness of the hyperfinite ${\rm II}_1$ factor among operator algebras, 
but also because of its unexpectedly rich representation theory among factors. 
Indeed, compared with others, the hyperfinite ${\rm II}_1$ factor has a 
proliferation of representations: the moduli space of its non-equivalent 
representations is isomorphic to $\R_+$ (and accordingly, the ${\rm II}_1$ 
factor is the only one among factors whose $K_0$ group is non-trivial, namely 
isomorphic to $\R$). However the existence of inequivalent 
representations, or in other words the failure of the {\it Stone--von 
Neumann representation theorem} in this case, is an indicator that the 
hyperfinite ${\rm II}_1$ factor is an operator algebra 
of a physical system possessing {\it infinitely many} degrees of 
freedom, like a (macroscopic) quantum statistical ensemble. As was mentioned 
by von Neumann, but with some amount of uncertainty, in {\it e.g.}  
\cite[Letters to P. Jordan, December 11, 1949 and January 12, 1950]{neu}, 
because of this property the hyperfinite ${\rm II}_1$ factor might even play 
a role in (relativistic) quantum field theory; although the recent 
conviction is that rather algebras of type III appear here, 
cf. \cite[Section V.6]{haa}.

The aim of this paper is to take two walks around representation theory 
of the hyperfinite ${\rm II}_1$ factor: one {\it mathematical} and one 
{\it physical} (of course these are not unrelated). The mathematical 
trip is a substantial extension of our previous efforts in 
\cite{ete3,ete4} and concerns the following problem: despite the 
existence of many non-trivial inequivalent representations of the 
hyperfinite ${\rm II}_1$ factor, only one of them appears as 
``reasonable'', namely its {\it standard representation}. Of course one 
can say that this is the most important while other representations, if 
cyclic, arise via the {\it Gelfand--Naimark--Segal 
construction} however using rather irrelevant states on the algebra, and the 
rest are direct sums of these; or all of them can be obtained uniformly by 
simply amplifying the standard representation. But one can also say that 
these general descriptions of representations are not too informative. Our 
first result towards constructing {\it contentful} new representations 
connects the general theory of the hyperfinite ${\rm II}_1$ factor with four 
dimensional differential geometry:

\begin{theorem} 
Let $M$ be a connected compact oriented smooth $4$-manifold.
Making use of its smooth structure only, out of $M$ a von Neumann algebra $\fr$
can be constructed which is geometric in the sense that it is generated by
geometric operators, including all complexified algebraic
({\it i.e.}, formal or stemming from a metric) curvature tensors of $M$. 
Moreover $\fr$ itself is a hyperfinite factor of type ${\rm II}_1$ hence is 
unique up to abstract isomorphisms of von Neumann algebras.

Furthermore $M$ admits an embedding $M\subset\fr$ via projections. Two 
$4$-manifolds $M,N$ with corresponding embeddings
have abstractly isomorphic von Neumann algebras however not canonically.
Nevertheless different abstract isomorphisms between these von Neumann
algebras induce orientation-preserving diffeomorphisms of $M$ and $N$
respectively {\it i.e.}, leave their embeddings unchanged. Hence up to
diffeomorphisms, all connected compact oriented smooth $4$-manifolds embed 
into a commonly given abstract von Neumann algebra $\fr$
which is the hyperfinite ${\rm II}_1$ factor.
\label{fotetel1}
\end{theorem}

\noindent The occurrence of the hyperfinite ${\rm II}_1$ factor in 
low dimensional differential topology is not only an immense source for 
new representations, but even brings a smooth $4$-manifold invariant to life:  

\begin{theorem}
Assuming $M$ to be as before, its von Neumann algebra $\fr$ 
admits a non-faithful representation on a certain complex separable Hilbert
space, such that the unitary equivalence class of this representation is
invariant under orientation-preserving diffeomorphisms of $M$. Consequently 
the Murray--von Neumann coupling constant of this representation
gives rise to a smooth $4$-manifold invariant $q$ taking values in the 
semi-open real interval $[0,1)$. Moreover there exists a subfactor 
$0\not=\fii(M)\subsetneqq\fr$ having index $[\fr\::\:\fii(M)]\in [4,+\infty)$ 
such that $q(M)=1-\frac{4}{[\fr\::\:\fii(M)]}$ and $q$ behaves like 
$q(M\#N)=q(M)+q(N)-q(M)q(N)$ under taking connected sum.
\label{fotetel2}
\end{theorem}

\noindent Let us make some general comments here. The spectrum of the
possible Jones indices looks like
$\big\{4\cos^2\big(\frac{\pi}{n}\big)\:\big\vert\:n=3,4,\dots\big\}
\cup[4,+\infty]$ hence splits into a discrete and a continuous part.
Subfactors with index in the discrete part
$\{4\cos^2\big(\frac{\pi}{n}\big)\:\vert\:n\geqq 3\}$ have been
completely classified \cite{pop} and in turn they follow an $ADE$
pattern \cite{ocn}. This classification scheme can be pushed further in
an affine form to cover the case when the index is exactly $4$ as well
\cite{jon-mor-sny}. However in general the set of subfactors with index
belonging to the continuous portion $[4,+\infty)$ is very wild and only
partial results are known mainly for the subinterval $[4,5]\subset
[4,+\infty)$ or a bit more (cf. {\it e.g.} \cite{jon-mor-sny} for an
excellent survey and recent results while for some further extension cf.
\cite{afz-mor-pen}). Concerning the impact of this division on the numerical
invariant $q$, the previous Theorem \ref{fotetel2} says that for a closed
smooth $4$-manifold $M$ its invariant always satisfies $\frac{4}{1-q(M)}\in
[4,+\infty)$ {\it i.e.}, the corresponding Jones index already belongs
to the continuous range. This observation is a hint that smooth
$4$-manifolds might provide a rich reservoir of exotic subfactors in the
wild {i.e.}, continuous index range. In fact, as we will see here, the
collection of connected closed orientable smooth $4$-manifolds apparently map
injectively into the collection of finite-index subfactors in the
continuous range of the hyperfinite ${\rm II}_1$ factor. Consequently
the difficult task of understanding smoothness phenomena in precisely
$4$ dimensions seems to be related with the also difficult problem of
taking an enumeration of subfactors of the hyperfinite ${\rm II}_1$
factor in the continuous-index range. 

Returning to the smooth invariant introduced in Theorem \ref{fotetel2}
it turns out that $q$ is not really sensitive in individual cases.
However, unlike other known smooth invariants which are only partially
defined and take discrete, typically integer values, $q$ is defined over
a large ensemble of $4$-manifolds and its image lies within a real
interval hence, when considered collectively, it exhibits certain
continuity properties. This allows us to make an interesting observation
namely that the global behaviour of $q$ over a hypothetical infinite
sequence of smooth $4$-manifolds carries information about a particular
but important case, the $4$-sphere. More precisely

\begin{theorem} Let $M$ be a connected closed orientable smooth $4$-manifold
as before but having finite fundamental group. Then
$q(M)=1-\vert\pi_1(M)\vert\big(\frac{2}{3}\big)^{2\:{\rm rk}(\pi_2(M))}$.

Furthermore if there exists a sequence $\{M_n\}_{n=1,2,\dots}$ of this sort of
$4$-manifolds whose $1^{\it st}$ and $2^{\it nd}$ homotopy groups behave
asymptotically like $\vert\pi_1(M_n)\vert\rightarrow+\infty$ and
$\left\vert\frac{\log\vert\pi_1(M_n)\vert}{\log\frac{9}{4}}-
{\rm rk}(\pi_2(M_n))\right\vert\rightarrow 0$
whenever $n\rightarrow+\infty$ then the $4$ dimensional smooth
Poincar\'e conjecture fails more precisely, the $4$-sphere admits
countable infinitely many inequivalent smooth structures. 
\label{fotetel3}
\end{theorem}

\noindent Intuitively speaking the smoothness issues over the $4$-sphere 
are ``converted'' into simpler topological properties of a particular 
infinite sequence of $4$-manifolds. An approach towards the 
(non-)existence of such a remarkable sequence of $4$-manifolds might be 
to consider finite group actions on simply connected closed 
$4$-manifolds, however this is beyond the scope of the present paper 
(but cf. \cite{edm}).

Next turning towards physics: a longstanding problem of contemporary 
theoretical physics is how to unify the obviously successful and 
mathematically consistent {\it theory of general relativity} with the 
obviously successful but yet mathematically problematic {\it 
relativistic quantum field theory}. It has been generally believed that 
these two fundamental pillars of modern theoretical physics are in a 
clash not only because of the different mathematical tools they use but 
are in tension even at a deep conceptual level: for instance classical 
notions of general relativity such as a space-time event, the light cone 
or the event horizon of a black hole are ``too sharp'' objects and the 
theory itself is ``too non-linear'' from a quantum theoretic viewpoint; 
whereas relativistic quantum field theory is not background independent 
from the aspect of general relativity.

The demand by general relativity summarized as the {\it principle of
general covariance} is perhaps one of the two main obstacles why general
relativity has remained outside of the mainstream classical and quantum
field theoretic expansion in the $20^{\rm th}$ century. Indeed, an
implementation of this inherent principle of general relativity forces that a
robust group, namely the full diffeomorphism group of the underlying space-time
manifold must belong to the symmetry group of a field theory
compatible with general relativity. However an unwanted consequence of
the vast diffeomorphism symmetry is that it even allows one to transform
time itself away from the theory (known as the ``problem of time'' in
general relativity, see {\it e.g.} \cite[Appendix E]{wal} for a technical
presentation as well as {\it e.g.} \cite[Chapter 2]{cal} and
\cite[Subsection 2.1]{hed2} for a broader philosophical survey on this
problem) making it problematic to apply usual canonical quantization 
methods---based on Hamiltonian
formulation hence on an essential explicit reference to an ''auxiliary
time''---in case of general relativity. The other reason is the as well
in-built core idea, the {\it equivalence principle} which renders
general relativity a strongly self-interacting classical field theory in
the sense that precisely in four dimensions the ``free'' and the
``interaction-with-itself'' modes of the gravitational field have
energetically the same magnitudes, obfuscating perturbative considerations.
In fact the equivalence principle says that there is no way to make a
physical distinction between these two modes of gravity. Heisenberg and
Pauli were still optimistic concerning canonical and perturbative
quantization of gravity with respect to a fixed time or, more generally, a 
reference or ambient space-time in their 1929 paper \cite{hei-pau}; however
these initial hopes quickly evaporated already in the 1930's by
recognizing the {\it essential impossibility} of quantizing general
relativity via canonical quantization and exhibit it as a perturbatively
renormalizable quantum field theory in a coherent way. This was clearly
observed by Bronstein \cite{bro} first; as he wrote in his 1936
paper: ``[...] the elimination of the logical inconsistencies [requires]
rejection of our ordinary concepts of space and time, modifying them by some
much deeper and nonevident concepts.'' (also cited by Smolin \cite[p. 85]{smo}).

Roughly the thinking about gravity has split into two main branches
since the 1950-60's \cite{hed2}. The first older and more accepted
direction postulates that gravity should be quantized akin to other
fundamental forces but with more advanced methods including
(super)string theoretic \cite{hed1}, Feynman integral, loop quantum gravity
or some further techniques---or at least one should construct it as a low
energy effective field theory of an unknown high energy theory; the other newer
and yet less-accepted attitude declares that gravity is an emergent
macroscopic phenomenon in the sense that it always involves a huge amount
of physical degrees-of-freedom (beyond the obvious astronomical
evidences, also supported by various theoretical discoveries during the
1970-80's such as Hawking's area theorem, black hole radiation, all
resembling thermodynamics) hence is not subject to quantization at all.
Nevertheless, as a matter of fact in the 2020's, we have to admit that
an overall accepted quantum theory of gravity does not exist yet and
even general relativity as a classical field theory persists to keep
its conceptually isolated position within current theoretical physics
\cite{hed2}. Perhaps it is worth mentioning here that general
relativity receives further challenges from low dimensional differential
topology too by recent discoveries which were unforeseeable earlier,
cf. {\it e.g.} \cite[Section 1]{ete1,ete2} for a brief summary.

Strongly motivated by these well-known general incompatibility comments, 
in the aforementioned second {\it i.e.}, physical trip around the hyperfinite 
${\rm II}_1$ factor, an operator algebraic characterization of the vacuum 
Einstein equation is obtained, which can be summarized as follows:

\begin{theorem} 
Let $M$ be a connected compact oriented smooth $4$-manifold and 
consider its embedding $M\subset\fr$ as in Theorem \ref{fotetel1}. This 
embedding induces a Riemannian structure $(M,g)$ whose Riemannian 
curvature satisfies $R_g\in\fr$. Moreover if $*$ denotes the 
Hodge star operating on $\Omega^2(M;\C)$ then $*\in\fr$. It is 
self-adjoint and satisfies $*^2=1$ hence is unitary thus generates a 
periodic inner $\divideontimes$-automorphism of $\fr$ 
rendering it a so-called {\em Hodge dynamical system} $\big(\fr,\big\{{\rm 
Ad}_{*^t}\}_{t\in\R}\big)$. Finally, both $M\subset\fr$ and the subfactor 
$\fii(M)\subset\fr$ from Theorem \ref{fotetel2} are 
preserved by the Hodge dynamics, more precisely they form parts of its 
fixed-point-subalgebra $\fii(M,g)\subseteqq\fr$, and $(M,g)$ is Einstein i.e., 
$R_g=\Lambda g$ if and only if $R_g$ belongs to this fixed-point-subalgebra 
too. 
\label{fotetel4} 
\end{theorem}

\noindent This result can be read as an indication to convert the real 
non-linear structure of general relativity, namely $(M,g, R_g)$ into 
a complex linear structure of quantum field theory, namely 
$\fii(M,g)\subseteqq\fr$.
 
The paper is organized as follows. Section \ref{two} contains detailed 
constructions leading to Theorem \ref{fotetel1} while Section 
\ref{three} exhbits the proofs of Theorems \ref{fotetel2} and \ref{fotetel3} 
through a sequence of separate lemmata. Finally Section \ref{four} is 
devoted to the proof of 
Theorem \ref{fotetel4} as well as placing representation theory of the 
hyperfinite ${\rm II}_1$ factor into the context of algebraic quantum field 
theory \cite{bra-rob,haa}.

%%%%%%%%%%%%%%%%%%%%%%%%%%%%%%%%%%%%%%%%%

\section{Emergence of the hyperfinite ${\rm II}_1$ factor} 
\label{two}

%%%%%%%%%%%%%%%%%%%%%%%%%%%%%%%%%%%

Since our earlier expositions \cite{ete3,ete4} are not widely-known as 
well as are yet technically unsatisfactory at some steps, in this 
section for the Reader's convenience we recall again the material in 
\cite{ete4}, however in a substantially modified and extended form, in 
order to prove Theorem \ref{fotetel1}. First we shall exhibit a new simple 
self-contained two-step construction of a von Neumann algebra attached 
to any compact oriented smooth $4$-manifold. Then the structure of this 
algebra will be explored in some detail.

{\it Construction of an algebra.} Take the isomorphism class of a connected 
compact oriented smooth $4$-manifold (without boundary) and from now on let 
$M$ be a once and for all fixed representative in it carrying the action of 
its own group of orientation-preserving diffeomorphisms ${\rm Diff}^+(M)$. 
Among all tensor bundles $T^{(p,q)}M$ over $M$ the $2^{\rm nd}$ exterior power 
$\wedge^2T^*M\subset T^{(0,2)}M$ is the only one which can be endowed with a 
pairing in a natural way {\it i.e.}, a pairing extracted from the smooth 
structure (and the orientation) of $M$ alone. Indeed, consider its associated 
vector space $\Omega^2(M):=C^\infty(M;\wedge^2T^*M)$ of smooth $2$-forms on 
$M$. Define an $L^2$-pairing $\langle\:\cdot\:, 
\:\cdot\:\rangle_{L^2(M)}:\Omega^2(M)\times \Omega^2(M) \rightarrow\R$ 
via integration: 
\begin{equation}
\langle\varphi,\psi\rangle_{L^2(M)}:=
\int\limits_M\varphi\wedge\psi
\label{integralas}
\end{equation}
and observe that this pairing is symmetric, non-degenerate however 
{\it indefinite} in general thus can be regarded as an indefinite 
$L^2$-scalar product on $\Omega^2(M)$. Therefore taking 
the complexified vector space $\Omega^2(M;\C):=
C^\infty(M;\wedge^2T^*M\otimes_\R\C)$ and the $\C$-bilinear extension of 
(\ref{integralas}), the assignment 
$\varphi\mapsto\langle\cdot\:,\varphi\rangle^\C_{L^2(M)}$ gives rise to a 
canonical inclusion $\Omega^2(M;\C)\subset(\Omega^2(M;\C))^*$ into the bare 
linear algebraic dual space. Consider $\ed(\Omega^2(M;\C))$, 
the unital algebra of all bare $\C$-linear endomorphisms of $\Omega^2(M;\C)$ 
and let $C(M)\subset \ed(\Omega^2(M;\C))$ denote the canonical complex 
subalgebra generated by the embedding $\Omega^2(M;\C)\subset(\Omega^2(M;\C))^*$ 
{\it i.e.}, $C(M)$ is generated by $\Omega^2(M;\C)\otimes_\C\Omega^2(M;\C)$ 
through the $\C$-linear inclusions $\Omega^2(M;\C)\otimes_\C\Omega^2(M;\C)
\subset\Omega^2(M;\C)\otimes_\C(\Omega^2(M;\C))^*\subset\ed(\Omega^2(M;\C))$. 
By definition elements of $C(M)$ look like 
$\omega\mapsto\sum\limits_\ell\langle\omega,\varphi_\ell
\rangle_{L^2(M)}^\C\psi_\ell$ such that {\it for every 
$\omega\in\Omega^2(M;\C)$} the right hand side {\it exists 
in $\Omega^2(M;\C)$}; in particular using 
any $\C$-linear basis $\{\varphi_1,\varphi_2,\dots\}$ satisfying 
$\langle\varphi_i,\varphi_j\rangle_{L^2(M)}^\C=\delta_{ij}$ then 
$\sum\limits_\ell\langle\omega,\varphi_\ell\rangle_{L^2(M)}^\C\varphi_\ell=
\omega$ shows that $C(M)$ is unital. For clarity note that being 
(\ref{integralas}) a non-local operation, $C(M)$ is a global infinite 
dimensional object, associated to the oriented smooth manifold $M$.

First we want to dissolve $C(M)$ into an infinite tensor product of matrix 
algebras. Start with an open covering of $M$ and by compactness 
select a finite open subcovering $\{U_i\}_{i=1,\dots,m}$ of it; we suppose 
that all $U_i$'s are already homeomorphic with open $4$-balls. Take a 
partition of unity $\{\rho_i\}_{i=1,\dots,m}$ subordinate to the covering 
({\it i.e.}, a collection of non-negative smooth functions satisfying 
$\sum\limits_{i=1}^m\rho_i=1$ such that the support of $\rho_i$ is contained in 
$U_i$). Put $\Omega^2_{\rho_i,0}(U_i):=\big\{\rho_i\psi_i\:\vert\:
\mbox{$\psi_i\in\Omega^2(U_i;\C)$ such that $\rho_i\psi_i
\vert_{\partial\overline{U}_i}=0$}\big\}$ 
and writing $U$ and $\rho$ for any of the $U_i$'s and their corresponding
$\rho_i$'s, let $C(U)$ be the algebra generated by 
$\Omega^2_{\rho,0}(U)\otimes_\C\Omega^2_{\rho,0}(U)$ as for $C(M)$ and 
consider $\fm_m\big(C(U)\big)$, the algebra of $m\times m$ matrices over 
$C(U)$. Picking $\rho_i\varphi_i,\rho_i\psi_i\in\Omega^2_{\rho_i,0}(U_i)$ 
and setting $\varphi:=\sum\limits_i\rho_i\varphi_i, \psi:=
\sum\limits_i\rho_i\psi_i\in\Omega^2(M;\C)$, 
the equality $\sum\limits_{i,j=1}^m\langle\:\cdot\:,\rho_i\varphi_i
\rangle^\C_{L^2(M)}\rho_j\psi_j=\langle\:\cdot\:,\varphi\rangle^\C_{L^2(M)}
\psi$ induces a homomorphism from 
$\fm_m(C(U))$ onto $C(M)$. This homomorphism is surjective however not 
injective: both follow from the associated decomposition 
$\varphi=\sum\limits_i(\rho_i\varphi)\mapsto\bigoplus
\limits_i(\rho_i\varphi)$ which is not 
unique due to overlappings in the open covering {\it i.e.}, because the open 
subset of intersections
\begin{equation}
\bigcup\limits_{i,j=1}^mU_i\cap U_j\subset M
\label{metszet}
\end{equation} 
is not empty. Consequently this surjective homomorphism has 
a non-trivial kernel $J(U)\subset\fm_m(C(U))$ yielding an isomorphism 
$C(M)\cong\fm_m(C(U))/J(U)$. Proceeding further, let us say that the open 
covering $\{V_k\}_{k=1,\dots,mn}$ of $M$ is an {\it $n$-homogeneous 
refinement} of $\{U_i\}_{i=1,\dots,m}$ if for every index 
$1\leqq k\leqq mn$ there exists $1\leqq i\leqq m$ satisfying $V_k\subset U_i$ 
and every $U_i$ contains 
precisely $n$ number of $V_k$'s which are again (smaller) open $4$-balls. Then 
repeating the previous decomposition procedure we obtain further that 
$C(M)\cong\fm_{mn}(C(V))/J(V)=\fm_m\big(\fm_n\big(C(V)\big)\big)/J(V)$ 
where $J(V)\subset\fm_{mn}(C(V))$ is the ideal generated by the 
overlappings in the refined covering $\{V_k\}_{k=1,\dots,mn}$ as before. 
By its definition as a topological space, $M$ admits a countable 
basis and any particular countable sequence of nested open subsets 
$M\supset U\supset V\supset\dots$ taken from the above refinement respectively, 
possesses the Cantor property {\it i.e.}, has a non-empty intersection; we can 
suppose without loosing generality that it is the particular point $x\in M$, 
and note that every point of $M$ arises in this manner. Therefore the induced 
sequence $C(M)\supset C(U)\supset C(V)\supset\dots$ of nested local unital 
algebras terminates at the corresponding full infinitesimal algebra 
$C(x)=\ed(\wedge^2T_x^*M\otimes_\R\C)\cong\fm_6(\C)$ 
since $\dim_\C(\wedge^2T_x^*M\otimes_\R\C)=\binom{4}{2}=6$. The 
corresponding sequence of ideals $0=J(M),J(U),J(V),\dots$ likewise 
terminates at $J(x)$. Geometrically speaking, since the points of $M$ 
are already disjoint hence the open subset (\ref{metszet}) of overlappings 
tends to the empty set during refinements, we get $J(x)=0$. Algebraically, 
since the ideals of $\fm_k(R)$ are in one one-to-one 
correspondence with the ideals of $R$ we can identify the ideals above with 
those within $C(M), C(U),C(V),\dots$ respectively. However 
since $C(x)\cong\fm_6(\C)$ has no non-trivial ideals, we find $J(x)=0$. 
Consequently the $\omega$-completed (in the sense of set theory) step-by-step 
procedure of taking homogeneous refinements of a finite open covering of 
$M$ formally induces a countable infinite resolution 
$C(M)\cong\fm_m\big(\fm_n(\dots(C(x))
\dots)\big)/J(x)\cong\fm_m\big(\fm_n(\dots(\fm_6(\C))
\dots)\big)\cong\fm_{mn\dots}(\C)$. This procedure can be captured 
precisely by the concept of an infinite tensor product: $C(M)$ arises as the 
inductive limit of matrix algebras $C_n(M)=\fm_{k_1\dots k_n}(\C)=
\fm_{k_1}(\C)\otimes_\C\dots\otimes_\C\fm_{k_n}(\C)$ where the 
embedding $C_n(M)\subset C_{n+1}(M)$ is by means of the mapping 
$A\mapsto A\otimes 1$ where $1\in\fm_{k_{n+1}}(\C)$ is the unit. 
To summarize: there exists a decomposition 
\begin{equation}
C(M)\cong\fm_{k_1}(\C)\otimes_\C\fm_{k_2}(\C) \otimes_\C\dots
\label{clifford}
\end{equation} 
yet depending on a particular step-by-step homogeneous 
refinement procedure of a particular finite covering of $M$; nevertheless 
by its construction $C(M)$ contains the infinitesimal hence finite dimensional 
algebras $\ed(\wedge^2T_x^*M\otimes_\R\C)$ at 
every $x\in M$ as well as forms part of the global algebra 
$\ed(\Omega^2(M;\C))$ yielding canonical inclusions 
\begin{equation} 
C^\infty(M;\ed(\wedge^2T^*M\otimes_\R\C))\subset C(M)\subset
\ed(\Omega^2(M;\C))=\ed\big(C^\infty(M;\wedge^2T^*M\otimes_\R\C)\big)
\label{izomorfizmus}
\end{equation} 
of $\C$-algebras by construction. 

Next we carry out completion. We see via (\ref{clifford}) that $C(M)$ is a 
complex $\divideontimes$-algebra whose $\divideontimes$-operation (provided by 
taking Hermitian matrix transpose, a non-local operation) is written as 
$A\mapsto A^\divideontimes$. The isomorphism (\ref{clifford}) also shows 
that if $A\in C(M)$ then one can pick the smallest $n\in\N$ such that 
$A\in C_n(M)$, consequently $A$ has a finite trace defined by 
$\tau (A):=\frac{1}{k_1k_2\dots k_n}{\rm Trace}(A)$ {\it i.e.}, taking the 
usual normalized trace of the corresponding $(k_1\dots k_n)\times 
(k_1\dots k_n)$ complex matrix. Observe that $\tau (A)\in\C$ does not depend on 
$n$. (We will see shortly that diffeomorphisms act by inner automorphisms on 
$C(M)$ hence by \cite[Corollary 2]{sto1} also $\tau: C(M)\rightarrow\C$ is 
the only state which is ${\rm Diff}^+(M)$-invariant {\it i.e.}, natural to use 
in our situation.) Following \cite[Section 1.6]{ana-pop} we then define a 
sesquilinear inner product on 
$C(M)$ by $(A,B):=\tau (AB^\divideontimes)$ which is non-negative and 
non-degenerate thus the completion of $C(M)$ with respect to the norm 
$\Vert\:\cdot\:\Vert$ induced by $(\:\cdot\:,\:\cdot\:)$ renders $C(M)$ a 
complex Hilbert space what we shall write as $\ch$ and its Banach algebra of 
all bounded linear operators as $\fb(\ch)$. Multiplication in $C(M)$ from 
the left on itself is continuous hence gives rise to a representation 
$\pi:C(M)\rightarrow\fb(\ch)$. Finally our key object effortlessly 
emerges as the weak closure of the image of $C(M)$ under $\pi$ within 
$\fb (\ch)$ or equivalently, by referring to von Neumann's bicommutant 
theorem \cite[Theorem 2.1.3]{ana-pop} we put 
\[\fr:=(\pi(C(M)))''\subset\fb(\ch)\:\:.\] 
This von Neumann algebra of course admits a unit $1\in\fr$ moreover 
continues to have trivial center {\it i.e.}, is a factor. Moreover by construction 
it is hyperfinite. The trace $\tau$ as defined extends from $C(M)$ to 
$\fr$ and satisfies $\tau (1)=1$. Moreover \cite[Proposition 4.1.4]{ana-pop} 
this trace is unique on $\fr$. Consequently $\fr$ is a 
hyperfinite ${\rm II}_1$ factor hence is unique up to abstract 
isomorphisms \cite[Theorem 11.2.2]{ana-pop}. In particular 
in this sense its construction is independent of the involved 
covering-refinement procedure over $M$. Going further, we also obtain by 
extension a representation 
$\pi:\fr\rightarrow\fb(\ch)$. Observe that here we have constructed 
$\ch$ as a completion of $(C(M),\tau)$ however the same $\ch$ arises if 
taking the completion of $(\fr,\tau)$. Hence the two kinds of 
completions $\fr$ and $\ch$ of one and the same object $C(M)$ in fact 
form an increasing chain $C(M)\subset\fr\subset\ch$ as complex 
(complete) vector spaces. Given $A,B\in\fr$ we 
shall write $A\in\fr$ but $\hat{B}\in\ch$ from now on as usual. This is 
necessary since $\fr$ and $\ch$ are very different for example as ${\rm 
U}(\ch)$-modules: given a unitary operator $V\in {\rm U}(\ch)$ then 
$A\in\fr$ is acted upon as $A\mapsto VAV^{-1}$ but $\hat{B}\in\ch$ 
transforms as $\hat{B}\mapsto V\hat{B}$. Using this notation the trace 
always can be written as a scalar product with the image of the unit in 
$\ch$ that is, for every $A\in\fr$ we have $\tau(A)=(\hat{A},\hat{1})$ 
exhibiting a general abstract expression for the trace.

{\it Exploring the algebra $\fr$.} Before proceeding further let us make 
a digression here to gain a better picture. This is desirable because 
taking the weak closure like $\fr$ of some explicitly known structure like 
$C(M)$, often involves a sort of loosing control over the latter. 
Nevertheless we already know promisingly that $\fr$ is a hyperfinite 
factor von Neumann algebra of ${\rm II}_1$ type. Let us now reveal some of its 
elements.
 
1. We immediately recognize via the left hand side of (\ref{izomorfizmus}) 
combined with the canonical inclusion $C(M)\subset\fr$ that 
$C^\infty (M;\ed(\wedge^2T^*M\otimes_\R\C))\subset\fr$ 
{\it i.e.}, the bundle or in other words pointwisely defined endomorphisms of 
$2$-forms form a part of $\fr$. As an introductory simple observation note that 
in case of such elements all operations in $\fr$ stem 
from the corresponding pointwise operations: if $a,b\in\C$ and 
$A,B\in C^\infty (M;\ed(\wedge^2T^*M\otimes_\R\C))$ therefore 
$A_x,A_x^*,B_x\in\ed(\wedge^2T_x^*M\otimes_\R\C)$ for every $x\in M$ then 
readily $(aA+bB)_x=aA_x+bB_x$, $(AB)_x=A_xB_x$ and $A^\divideontimes_x=A^*_x$. 
Operators of this kind are important because they 
allow to make a contact with local {\it four dimensional} differential 
geometry.\footnote{In fact all the constructions so far work for an arbitrary 
compact connected oriented and smooth $4k$-manifold with $k=0,1,2,\dots$ 
(note that in $4k+2$ dimensions the pairing (\ref{integralas}) 
gives rise to a symplectic structure on $2k+1$-forms).\label{4k}} 
A peculiarity of four dimensions is that the space 
$C^\infty(M;\ed(\wedge^2T^*M\otimes_\R\C ))$ contains 
curvature tensors (more precisely their complex linear extensions) 
on $M$. If $(M,g)$ is an oriented Riemannian $4$-manifold 
then its Riemannian curvature tensor $R_g$ is indeed a member of this 
subalgebra: with respect to the splitting of complexified $2$-forms into 
their (anti)self-dual parts its complex linear extension looks like 
(cf. \cite{sin-tho})
\begin{equation}
R_g=\left(\begin{matrix}
                \frac{1}{12}{\rm Scal}+{\rm Weyl}^+ & {\rm Ric}_0\\
                  {\rm Ric}_0^* & \frac{1}{12}{\rm Scal}+{\rm Weyl}^-
                      \end{matrix}\right)
\:\:\:\:\::\:\:\:\:\:\begin{matrix}\Omega^+(M;\C)\\
                                                  \bigoplus\\
                                                 \Omega^-(M;\C)
                                            \end{matrix}
         \:\:\:\longrightarrow\:\:\:\begin{matrix}\Omega^+(M;\C)\\
                                                  \bigoplus\\
                                                 \Omega^-(M;\C)
                                            \end{matrix}
\label{gorbulet}
\end{equation}
hence is a self-adjoint operator $R^\divideontimes_g=R_g$. 
More generally, $C^\infty(M;\ed(\wedge^2T^*M\otimes_\R\C ))$ hence $\fr$ 
contains the complex linear extensions of all algebraic ({\it i.e.}, 
formal only, not stemming from a metric) curvature tensors $R$ over $M$. 

Whenever $A\in C^\infty (M;\ed (\wedge^2T^*M\otimes_\R\C))$ one can 
work out an explicit formula for its trace because one can compare 
the global trace $\tau(A)$ and the local trace function 
$x\mapsto {\rm tr}(A_x)$ given by the pointwise traces of the local operators 
$A_x:\wedge^2T^*_xM\otimes_\R\C\rightarrow\wedge^2T^*_xM\otimes_\R\C$ at every 
$x\in M$. Recall that $\fr$ has been constructed as the weak closure of 
the algebra (\ref{clifford}). In fact \cite[Section 1.1.6]{ana-pop} the 
universality of $\fr$ permits to obtain it from any particular sequence of 
matrix algebras, like {\it e.g.} $\fm_6(\C)\otimes_\C\fm_6(\C)\otimes_\C...$ 
whose weak closure therefore is again $\fr$. Choose an arbitrary Riemannian 
metric $g$ on the compact oriented $M$, extend it sesquilinearly and 
take the corresponding definite $L^2$-scalar product on $\Omega^2(M;\C)$. 
That is, put 
\[(\varphi,\psi)^\C_{L^2(M,g)}:=\int\limits_M\varphi\wedge 
*_g\overline{\psi}=\int\limits_{x\in M}g_x(\varphi_x,\psi_x)^\C\mu_g(x)\] 
as usual, where $\mu_g\in\Omega^4(M)$ is the induced volume form on $(M,g)$. 
Pick an orthonormal frame $\{\varphi_1,\varphi_2,\dots\}$ in $\Omega^2(M;\C)$. 
By the aid of these tools the trace of $A$ takes the shape
\[\tau (A)=\lim\limits_{n\rightarrow+\infty}\frac{1}{6^n}
\sum\limits_{i=1}^{6^n} (A\varphi_i\:,\:\varphi_i)^\C_{L^2(M,g)}\:\:.\] 
Fix $n\in\N$, write 
$M_n:=\bigcap\limits_{i=1}^{6^n} {\rm supp}\varphi_i\subseteqq M$ and 
take a point $x\in M_n$. Since 
$\dim_\C(\wedge^2T^*_xM\otimes_\R\C)=\binom{4}{2}=6$ the maximal number of 
completely disjoint linearly independent sub-$6$-tuples in 
$\{\varphi_{1,x},\varphi_{2,x},\dots,\varphi_{6^n,x}\}$ is equal to 
$\frac{6^n}{6}=6^{n-1}$. Moreover it follows from Sard's lemma that with 
a generic smooth choice for $\{\varphi_1,\varphi_2,\dots\}$ the subset 
of those points $y\in M_n$ where this number is less than $6^{n-1}$ has 
measure zero in $M_n$ with respect to the measure $\mu_g$. Consequently 
\[\sum\limits_{i=1}^{6^n}(A\varphi_i,\varphi_i)^\C_{L^2(M_n,g)}= 
\int\limits_{x\in M_n}\sum\limits_{i=1}^{6^n}g_x(A_x\varphi_{i,x}, 
\varphi_{i,x})^\C\:\mu_g(x)=6^{n-1}\int\limits_{x\in M_n}{\rm tr} 
(A_x)\mu_g(x)\:\:.\] 
Since $\{\varphi_1,\varphi_2,\dots\}$ is a basis in $\Omega^2(M;\C)$ 
therefore $M\setminus\bigcup\limits_{n=0}^{+\infty}M_n=
\bigcap\limits_{n=0}^{+\infty}(M\setminus M_n)$ has measure zero as 
well we can let $n\rightarrow+\infty$ to end up with 
\[\tau (A)=\frac{1}{6}\int\limits_M{\rm tr}(A)\mu_g\:\:.\] 
Notice by inserting (\ref{gorbulet}) that 
$\tau(R_g)=\frac{1}{12}\int_M{\rm Scal}\:\mu_g$ hence the operator algebraic 
trace on the associated von Neumann algebra 
is nothing else than the generalization of the total scalar 
curvature in Riemannian geometry (or the Einstein--Hilbert action 
of vacuum general relativity).

2. Next we assert that the group of invertible elements within 
$C(M)$ can be canonically identified with 
${\rm Iso}(\wedge^2T^*M\otimes_\R\C)$, the group of bundle isomorphisms 
{\it i.e.}, pairs $(\Phi, F)$ consisting of an orientation-preserving 
diffeomorphism $\Phi$ of the base $M$ and a fiberwise $\C$-linear 
diffeomorphism $F$ of the total space $\wedge^2T^*M\otimes_\R\C$ such that
\[\xymatrix{\ar[d]_\pi\wedge^2T^*M\otimes_\R\C\ar[r]^{F} &
             \wedge^2T^*M\otimes_\R\C\ar[d]_\pi \\
          M\ar[r]^{\Phi}  & M}\]
is commutative. On the one hand by the aid of the right hand side inclusion in
(\ref{izomorfizmus}) invertible elements in $C(M)$ belong to 
$\ed(\Omega^2(M;\C))$; on the other hand it is clear that smooth bundle 
isomorphisms act on the space of smooth sections in an obvious $\C$-linear way 
hence ${\rm Iso}(\wedge^2T^*M\otimes_\R\C)\subset\ed(\Omega^2(M;\C))$ too.
Consequently our assertion at least makes sense. Moreover intuitively 
it also follows at once, for the very geometric content of 
the decomposition (\ref{clifford}) in its $\omega$-completed form is 
that elements of $C(M)$, while operating on smooth sections, 
preserve their infinitesimal structure {\it i.e.}, the individual fibers of 
$\wedge^2T^*M\otimes_\R\C$ too, hence the invertible ones give rise to 
bundle morphisms. 

More precisely, by recalling again the construction leading to the 
decomposition (\ref{clifford}), take a finite open covering 
$\{U_i\}_{i=1,\dots,m}$ of $M$, pick an invertible operator 
$A\in C(M)=\fm_m(C(U))/J(U)$ and consider 
any of its preimage $A_m\in\fm_m(C(U))$. This is an $m\times m$ invertible 
matrix consequently admits a unique ``$LU$-decomposition'' {\it i.e.}, 
$A_m=T'_m\Pi_mT''_m$ such that $T'_m,T''_m\in\fm_m(C(U))$ are upper 
triangular matrices moreover $\Pi_m\in\fm_m(C(U))$ is a permutation matrix 
(this is meaningful here because $0,1\in C(U)$ always); let $\pi_m$ be the 
corresponding permutation of the index set $\{1,2,\dots,m\}$. In this way, 
on the one hand, $A_m$ induces an approximate diffeomorphism of 
$\{U_i\}_{i=1,\dots,m}$ by putting $\Phi_m(U_i):=U_{\pi_m(i)}$ for every 
$i=1,\dots,m$. On the other hand, $A_m$ gives rise to an approximate element 
$F_m$ (by considering its action on very localized $2$-forms), which is 
compatible with $\Phi_m$. These together comprise an approximating bundle 
isomorphism $(\Phi_m,F_m)$. Then taking a countable sequence of step-by-step 
homogeneous refinements as before we obtain a bundle isomorphism $(\Phi,F)$ 
assigned to $A$. Conversely, given a pair $(\Psi,G)$, on the one hand take an 
open covering $\{V_j\}_{j=1,\dots,n}$ which is invariant under $\Psi$. On 
the other hand viewing $G$ as an invertible $\C$-linear transformation of 
$\Omega^2(M;\C)$ consider the matrix $B_n\in\fm_n(C(V))$ built up from the 
action of $G$ using localized $2$-forms belonging to 
$\Omega^2_{\rho_j,0}(V_j)$ for every $j=1,\dots,n$ that is, define 
$(B_n)_{ij}$ by solving 
$G=\langle\:\cdot\:,\rho_i\psi_i\rangle^\C_{L^2(M)}\rho_j\psi_j$ 
with suitable local $2$-forms, and take the image of this matrix in 
$C(M)=\fm_n(C(V))/J(V)$. Then carrying out 
refinements as usual we obtain an invertible operator $B\in C(M)$. 

Note that the identification of invertible 
elements of $C(M)$ with bundle isomorphisms has already been partly 
recognized in the previous item, for the left hand side embedding of 
(\ref{izomorfizmus}) implies that the gauge group 
$\cg(\wedge^2T^*M\otimes_\R\C)=C^\infty (M;{\rm Aut}(\wedge^2T^*M
\otimes_\R\C))$ of $\wedge^2T^*M\otimes_\R\C$ embeds into the invertible 
part of $C(M)$. Indeed, there exists a short exact sequence 
\[1\longrightarrow \cg(\wedge^2T^*M\otimes_\R\C)
\longrightarrow{\rm Iso}(\wedge^2T^*M\otimes_\R\C)
\longrightarrow {\rm Diff}^+(M)\longrightarrow 1\]
which injects the gauge group into the isomorphism group. 
Moreover $\Phi\mapsto(\Phi^{-1})^*$ as a bundle pullback gives a group 
injection ${\rm Diff}^+(M)\subset{\rm Iso}(\wedge^2T^*M\otimes_\R\C)$ too 
whose composition with the projection 
${\rm Iso}(\wedge^2T^*M\otimes_\R\C)\rightarrow {\rm Diff}^+(M)$ is the 
identity. Consequently this sequence can be supplemented to  
\begin{equation}
1\longrightarrow \cg(\wedge^2T^*M\otimes_\R\C)
\longrightarrow{\rm Iso}(\wedge^2T^*M\otimes_\R\C)\:\:^
{\longrightarrow}_{\longleftarrow}\:\:{\rm Diff}^+(M)\longrightarrow 1
\label{hasadas}
\end{equation}
implying a splitting ${\rm Iso}(\wedge^2T^*M\otimes_\R\C))=
\cg(\wedge^2T^*M\otimes_\R\C)\rtimes{\rm Diff}^+(M)$ as a
semi-direct product. Therefore ${\rm Diff}^+(M)\subset C(M)$ too, and 
$\Phi\in{\rm Diff}^+(M)$ acts on $\Omega^2(M;\C)$ by the corresponding 
pullback $(\Phi^{-1})^*$ on $2$-forms which is a $\C$-linear operator. 
Composing this with $C(M)\subset\fr$ we find that 
${\rm Diff}^+(M)\subset\fr$ {\it i.e.}, the group of orientation-preserving 
diffeomorphisms form a part of $\fr$ too. 

Furthermore, it is easy to see that ${\rm Diff}^+(M)\subset\fr$ acts by 
unitaries on $\ch$ via the standard representation 
$\pi:\fr\rightarrow\fb(\ch)$. Using any definite $L^2$-scalar product on 
$\Omega^2(M;\C)$ and a corresponding base $\{\varphi_1,\varphi_2,\dots\}$ as 
above, we can represent $(\Phi^{-1})^*$ by a countable infinite 
complex matrix according to (\ref{clifford}) whose Hermitian conjugate 
$((\Phi^{-1})^*)^\divideontimes$ satisfies 
$\big(\varphi,(\Phi^{-1})^*\psi\big)_{L^2(M,g)}=
\big(((\Phi^{-1})^*)^\divideontimes\varphi,\psi\big)_{L^2(M,g)}$. But the 
diffeomorphism-invariance of the $L^2$-scalar product implies     
\[\big(((\Phi^{-1})^*)^\divideontimes\varphi,\psi\big)^\C_{L^2(M,g)}=
\big(\varphi,(\Phi^{-1})^*\psi\big)^\C_{L^2(M,g)}=
\big(\Phi^*\varphi, \Phi^*(\Phi^{-1})^*\psi\big)^\C_{L^2(M,g)}=
\big(\Phi^*\varphi,\psi\big)^\C_{L^2(M,g)}\]
demonstrating $((\Phi^{-1})^*)^\divideontimes=\Phi^*$ that is,
diffeomorphisms are unitary over $\big(\Omega^2(M;\C)\:,
\:(\:\cdot\:,\:\cdot\:)_{L^2(M,g)}\big)$. Consider now the dense subset 
of $\ch$ having elements of the form $\hat{A}\in\ch$ where 
$A\in C(M)\subset\fr$. Then taking the scalar product 
$(\:\cdot\:,\:\cdot\:)$ of $\ch$ and 
exploiting the cyclic property of the trace as well as the already recognized 
unitarity of $\Phi$ on $\big(\Omega^2(M;\C)\:,
\:(\:\cdot\:,\:\cdot\:)_{L^2(M,g)}\big)$ we find that 
\begin{eqnarray}
\big(\pi((\Phi^{-1})^*)\hat{A},\pi((\Phi^{-1})^*)\hat{B}\big)&=&\tau\big(
(\Phi^{-1})^*A((\Phi^{-1})^*B)^\divideontimes\big)=
\tau\big(((\Phi^{-1})^*)^\divideontimes (\Phi^{-1})^*AB^\divideontimes\big)=
\tau(AB^\divideontimes)\nonumber\\
&=&\big(\hat{A},\hat{B}\big)\nonumber
\end{eqnarray}
proving that the embedding ${\rm Diff}^+(M)\subset\fr$ provided by 
$\Phi\mapsto(\Phi^{-1})^*$ is indeed unitary. This also demonstrates in 
particular that $\vert\tau((\Phi^{-1})^*)\vert=1$. 

3. Last but not least we exhibit an especially important class of 
elements in $\fr$ which permit to inject $M$ back into its $\fr$ as 
follows (also cf. \cite{ber-bes-gal}). First write $M':=M$ and put 
$P_{M'}:=0\in C(M)=\fm_1(C(M))$. Next, 
take a finite open covering $\{U_i\}_{i=1,\dots,m}$ of $M$ and its associated 
amplified algebra $\fm_m(C(U))$ as before. After re-indexing if necessary, 
consider the open subsets $U:=U_1$ and $U':=U_1\cup U_2\cup\dots\cup U_{m_1}$ 
with $1\leqq m_1\leqq m$ such that $U_1\cap U_i\not=\emptyset$ whenever 
$1\leqq i\leqq m_1$ ({\it i.e.}, take all subsets $U_i$ having non-empty 
intersections with the distinguished open subset $U_1$ in the covering). 
Note that $U\Subset U'$ ({\it i.e.}, $\overline{U}\subset U'$). Picking the 
local zero and unit $0,1\in C(U)$ consider the diagonal matrix $P_{U'}
%:=\left(\begin{smallmatrix}
%                                             0 & 0 & \cdot & 0 & 0\\
%                                             0 & 0 & \cdot & 0 & 0\\
%                                         \cdot & \cdot & \cdot & \cdot&\cdot\\ 
%                                             0 & 0 & \cdot & 1 & 0\\
%                                             0 & 0 & \cdot & 0 & 1
%                                        \end{smallmatrix}\right)
                                          \in\fm_m(C(U))$ containing $0$'s 
in the first $m_1$ entries and $1$'s in the last $m-m_1$ entries along its 
diagonal. This matrix maps onto a non-zero element in 
$C(M)\cong\fm_m(C(U))/J(U)$ what we continue to write as $P_{U'}$. 
As an operator on smooth $2$-forms, it acts like 
$P_{U'}\varphi=P_{U'}\Big(\sum\limits_{i=1}^m\rho_i\varphi\Big)=
\sum\limits_{i=m_1+1}^m\rho_i\varphi$ hence is a projection onto the subspace 
$\Omega^2(M,\overline{U'};\C)\subset\Omega^2(M;\C)$ of $2$-forms vanishing on
the closed subset $\overline{U'}\subseteqq M$. Note however that, 
considered as $P_{U'}\in\fr$ via $C(M)\subset\fr$, this operator 
is already an abstract projection too, yet characterized by its 
geometric behaviour that preserves $\Omega^2(M,\overline{U'};\C)$. 
Conversely, in fact any projection in $\fr$ possessing this behaviour can be 
taken to play the role of a $P_{U'}$ here. Repeating the construction for a 
countable sequence of nested open subsets and their surrounders 
\[\begin{matrix}
M&\supset &U&\supset &V&\supset&\dots\\
\rotatebox[origin=c]{90}{$=$}& &\cap & &\cap &&\\
M'&\supset &U'&\supset &V'&\supset&\dots
\end{matrix}\] 
collapsing to some point $x\in M$ during a countable step-by-step 
homogeneous refinement of the open covering as before, and referring to 
the decomposition (\ref{clifford}) as well as using $C(M)\subset\fr$, 
in this way we come up with an element $P_x\in\fr$. This is an abstract 
projection of $\fr$ however note that $P_x$ inherits the geometric property 
that it acts as the identity precisely on $\Omega^2(M,x;\C)$ as well. 
Hence surely $0\not=P_x\not=1$ and $P_x\not=P_y$ for every $x,y\in M$ such 
that $x\not=y$. We obtain therefore a projection $P_x\in\fr$ assigned 
to the point $x\in M$ and the resulting map 
\begin{equation}
i_M:M\longrightarrow\fr
\label{beagyazas}
\end{equation}
defined by $x\mapsto P_x$ is obviously injective. 

If $x,y\in M$ take a smooth $1$-parameter subgroup 
$\{\Phi_t\}_{t\in[-1,+1]}\subset{\rm Diff}^+(M)$ such that the curve 
$t\mapsto\Phi_t(x)$ connects $x$ with $y$ {\it i.e.}, $\Phi_0(x)=x$ and 
$\Phi_1(x)=y$ or equivalently $\Phi_1^{-1}(y)=\Phi_{-1}(y)=x$ consequently 
for the corresponding subspaces $(\Phi_{-1})^*\Omega^2(M,x;\C)=
\Omega^2(M,y;\C)$. Consider the induced family of projections 
$P(t):=\Phi_{-t}^*P_x\Phi^*_t$ within $\fr$. Then $x\mapsto P_x$ is 
continuous in the sense that $P(0)=P_x$ and $P(1)=P_y$. In particular all of 
these projections are unitary equivalent hence $0<\tau(P_x)=\tau(P_y)<1$ for 
every $x,y\in M$. It also follows that the image of $M$ in $\fr$ is always 
closed in any usual topology on $\fr$. Indeed, let $x_1,x_2,\dots$ be a 
countable sequence 
of points in $M$ such that the corresponding series of projections 
$P_{x_1},P_{x_2},\dots$ in $\fr$ converges to a unique operator $A\in\fr$. Then 
there exist (non-unique) elements $\Phi_1,\Phi_2,\dots$ in ${\rm Diff}^+(M)$ 
such that $x_i=\Phi_i(x_1)$ hence $P_{x_i}=(\Phi_i^{-1})^*P_{x_1}\Phi^*_i$. By 
compactness of $M$ there is a (sub)sequence converging to a point $x\in M$ as 
$j\rightarrow+\infty$ and a (non-unique) transformation $\Phi$ such that 
$\Phi(x_1)=x$. We can suppose that 
$\Phi_{x_j}\rightarrow\Phi$ in the induced subspace topology on 
${\rm Diff}^+(M)\subset\fr$. Regarding the projections 
$P_{x_j}=(\Phi_j^{-1})^*P_{x_1}\Phi^*_j\rightarrow(\Phi^{-1})^*P_{x_1}\Phi^*=
P_x$ as $j\rightarrow+\infty$ thus the original limit 
operator satisfies that $A=P_x$ with some unique $x\in M$. 
We conclude that (\ref{beagyazas}) in any usual topology on $\fr$ gives 
rise to a closed embedding of $M$ into $\fr$ via projections.  

4. Having understood (\ref{beagyazas}) for a given 
$4$-manifold, let us compare them for different spaces. So 
let $M,N$ be two connected compact oriented smooth $4$-manifolds and 
consider their corresponding embeddings $M\subset\fr$ and $N\subset\fr$ into 
their von Neumann algebras via (\ref{beagyazas}) respectively. Regardless 
what $M$ or $N$ are, their abstractly given algebras are both 
hyperfinite factors of ${\rm II}_1$ type, therefore these latter objects are 
isomorphic \cite[Theorem 11.2.2]{ana-pop} however not in a 
canonical fashion. Indeed, if $F':\fr\rightarrow\fr$ is an abstract 
isomorphism between the $\fr$'s for $M$ and $N$ respectively, then any other 
abstract isomorphism between them has the form $F''=\beta^{-1}F'\alpha$ where 
$\alpha$ and $\beta$ are $\C$-linear $\divideontimes$-automorphisms (in short 
from now on: automorphisms) of the abstractly given $\fr$ for $M$ and the 
abstractly given $\fr$ for $N$ respectively.

Therefore to understand the freedom how operator algebras for different 
$4$-manifolds are identified we have to understand automorphisms of $\fr$. 
First, let us see how inner automorphisms of the weakly dense subalgebra 
$C(M)$, which are of course conjugations with its unitary elements, 
look like. It follows from the splitting (\ref{hasadas}) of the group of
invertible elements 
that an inner automorphism of $C(M)$ admits a unique decomposition 
${\rm Ad}_\gamma{\rm Ad}_{(\Phi^{-1})^*}$ where 
$\gamma\in\cg(\wedge^2T^*M\otimes_\R\C)$ is a $\C$-linear 
gauge transformation of the bundle $\wedge^2T^*M\otimes_\R\C$ hence leaves $M$ 
pointwise fixed, and $\Phi\in{\rm Diff}^+(M)$ is a diffeomorphism of $M$ 
which also preserves $M$ as a whole. It then readily follows that the 
image of $M$ within $\fr$ given by (\ref{beagyazas}) is preserved by this 
automorphism. Next, consider the strong$^\divideontimes$-topology on 
$\fb(\ch)$. Taking $\fr\subset\fb(\ch)$ it is clear that $C(M)\subset\fr$ is 
a dense subalgebra; moreover also taking the unique unitary implementation 
${\rm Aut}\:\fr\subset\fb(\ch)$ of the full automorphism group 
\cite[Proposition 7.5.2]{ana-pop}, it is known 
\cite[Theorem 4]{sak2} that the subgroup of inner automorphisms is also 
dense within ${\rm Aut}\:\fr$. Consequently ${\rm Inn}\:C(M)$ itself is 
strong$^\divideontimes$-dense in ${\rm Aut}\:\fr$ hence any 
automorphism $\alpha$ of $\fr$ arises as   
\begin{equation}
\alpha=\lim\limits_i\big({\rm Ad}_{\gamma_{\:i}}{\rm Ad}_{(\Phi^{-1}_i)^*}\big)
\label{automorfizmus}
\end{equation}
where the limit is taken in the strong$^\divideontimes$-topology on $\fb(\ch)$. 
Therefore, taking into account as well that $M$ is a 
strong$^\divideontimes$-closed subset of $\fr$ under (\ref{beagyazas}), we 
conclude that $\alpha$ also preserves $M\subset\fr$. Likewise, 
${\rm Ad}_\delta{\rm Ad}_{(\Psi^{-1})^*}$ is the shape of an inner 
automorphism of $C(N)$ and a generic automorphism $\beta$ arises as 
strong$^\divideontimes$-limit of them. These make 
sure that given an abstract isomorphism $F':\fr\rightarrow\fr$ between the von 
Neumann algebras constructed for $M$ and $N$ respectively, then any other 
abstract isomorphism can be expressed as 
$F''=\lim\limits_i\big({\rm Ad}_{\Psi_i^*}({\rm Ad}_{\delta_i^{-1}}F'
{\rm Ad}_{\gamma_{\:i}}){\rm Ad}_{(\Phi_i^{-1})^*}\big)$ between them. 
Consequently abstract isomorphisms between a pair of abstractly given von 
Neumann algebras differ only by automorphisms which preserve 
their underlying $4$-manifolds as embedded within their algebras via 
(\ref{beagyazas}) respectively; that is, differences between identifications 
are inessential in this sense. Our overall conclusion here therefore is that 
{\it up to diffeomorphisms every connected compact oriented smooth 
$4$-manifold $M$ admits a closed embedding into a commonly given abstract 
von Neumann algebra $\fr$ via (\ref{beagyazas})}. 

5. We close the partial comprehension of $\fr$ with an observation regarding 
its general structure. Recall that the group of invertible 
elements of $C(M)$ decomposes as the semi-direct product of the group of 
gauge transformations of $\wedge^2T^*M\otimes_\R\C$ and the group of 
diffeomorphisms of $M$ via (\ref{hasadas}); thus via $C(M)\subset\fr$ a 
weakly dense subset of $\fr$ therefore admits a decomposition into bundle 
endomorphisms of $\wedge^2T^*M\otimes_\R\C$ and diffeomorphisms of $M$. 
Consequently given a $4$-manifold $M$, its associated von Neumann algebra 
$\fr$ is geometric in the sense that it is weakly generated by geometric 
operators.  

{\it Summing up our findings so far.} Given a connected compact 
oriented smooth $4$-manifold $M$ there exists a hyperfinite factor von 
Neumann algebra of ${\rm II}_1$ type $\fr$ associated to $M$ such that the 
solely input in its construction has been the pairing (\ref{integralas}). 
Hence $\fr$ depends only on the orientation and the smooth structure of $M$. It 
contains, certainly among many other non-geometric operators, the space 
$M$ itself as projections, its orientation-preserving diffeomorphisms as 
well as the space of algebraic curvature tensors. Nevertheless 
$\fr$ is geometric in the sense that it is generated by $M$'s 
operators alone. It is remarkable that despite the {\it plethora} of 
smooth $4$-manifolds detected since the early 1980's their associated 
von Neumann algebras here are unique offering a sort of justification 
terming $\fr$ as ``universal''. Moreover one is permitted to say that 
every connected oriented smooth $4$-manifold $M$ (perhaps together with 
its curvature tensor) embeds up to diffeomorphisms 
hence in a functorial way into a {\it common} $\fr$, 
and to look upon this von Neumann algebra as a natural common non-commutative 
space generalization of all oriented smooth $4$-manifolds (or all 
$4$-geometries). This universality also justifies the simple notation $\fr$ 
used throughout the text. 
\vspace{0.1in}

\noindent{\it Proof of Theorem \ref{fotetel1}}: Putting together all 
considerations so far the theorem is proved.\hspace{2.5cm}$\square$

%%%%%%%%%%%%%%%%%%%%%%%%%%%%%%%%%%%%%%%%%

\section{Some representations of the hyperfinite ${\rm II}_1$ factor}
\label{three}

%%%%%%%%%%%%%%%%%%%%%%%%%%%%%%%

\noindent{\it Representations of $\fr$ and a new smooth $4$-manifold 
invariant.} We can now proceed further: the next lemmata closely follow 
\cite[Lemmata 2.1-2.4]{ete4} but with substantially improved constructions 
and extended contents.

\begin{lemma} Let $M$ be a connected compact oriented smooth $4$-manifold 
and $\fr$ its von Neumann algebra with trace $\tau$ as before. Then there 
exists a complex separable Hilbert space $\ci (M)^\perp$ and a representation 
$\rho_M:\fr\rightarrow\fb (\ci (M)^\perp)$ with the following properties. 
If $\pi:\fr\rightarrow\fb (\ch)$ is the standard representation constructed 
above then $0\subseteqq\ci (M)^\perp\subsetneqq\ch$ and 
$\rho_M=\pi\vert_{\ci (M)^\perp}$ holds. 

Moreover the unitary equivalence class of $\rho_M$ is invariant under 
orientation-preserving diffeomorphisms of $M$. Thus the Murray--von Neumann 
coupling constant\footnote{Also called the {\it $\fr$-dimension} of a 
left $\fr$-module hence denoted $\dim_{\fr}$, cf. 
\cite[Chapter 8]{ana-pop}.\label{dimenzio}} of $\rho_M$ is invariant under 
orientation-preserving diffeomorphisms. Writing 
$Q_M:\ch\rightarrow\ci(M)^\perp$ for the orthogonal projection and taking into 
account the characterization of $\ci(M)^\perp$ the coupling constant is equal 
to $\tau(Q_M)\in[0,1)$ hence $q(M):=\tau(Q_M)$ is a 
smooth $4$-manifold invariant such that $q(M)\in[0,1)$ holds.
\label{terbelilemma}
\end{lemma}
 
\begin{proof} First let us exhibit a representation of $\fr$; we proceed 
from the Gelfand--Naimark--Segal construction however relative to the 
already existing standard representation. As we have 
seen every connected compact oriented smooth $4$-manifold $M$ admits an 
embedding into its von Neumann algebra $\fr$ via (\ref{beagyazas}) and this 
map looks like $x\mapsto P_x$ {\it i.e.}, to every point $x\in M$ the map 
$i_M:M\rightarrow\fr$ assigns a projection $P_x\in\fr$ which always 
satisfies that $0\not=P_x\not=1$. Take any 
Riemannian metric $g$ on $M$ having volume form $\mu_g$ over $M$. 
Then consider, as a natural choice, the map $F_{M,g}:\fr\rightarrow\C$ given by 
\[F_{M,g}(A):=\int\limits_{x\in M}i_M^*\tau\big(AP_x\big)\mu_g(x)\] 
which is well-defined due to compactness of $M$. This map is obviously 
$\C$-linear, ultraweakly continuous and satisfies 
$F_{M,g}(A^\divideontimes)=\overline{F_{M,g}(A)}$. Moreover we compute 
$F_{M,g}(A^\divideontimes A)=\int_Mi_M^*\big\Vert\widehat{AP_x}
\big\Vert^2\mu_g(x)\geqq 0$ consequently $F_{M,g}$ is a non-negative 
functional on $\fr$ and in particular $F_{M,g}(1^\divideontimes1)>0$. 
Therefore an application of the standard inequality 
$\vert F_{M,g}(A^\divideontimes B)\vert^2
\leqq F_{M,g}(A^\divideontimes A)F_{M,g}(B^\divideontimes B)$ implies that 
$0\subseteqq I(M,g)\subseteqq\fr$ consisting of
those operators $B\in\fr$ for which $F_{M,g}(B^\divideontimes B)=0$ 
is a multiplicative left-ideal in $\fr$ such that 
$\C^\times1\not\subset I(M,g)$ thus surely 
$0\subseteqq I(M,g)\subsetneqq\fr$. In fact $I(M,g)$ is independent of 
the metric $g$ involved in its definition: if $h$ is another Riemannian metric 
on $M$ then there exists a positive function $f:M\rightarrow\R$ such that 
\[0\leqq F_{M,h}(A^\divideontimes A)=
\int\limits_{x\in M}i_M^*\tau\big(A^\divideontimes AP_x\big)\mu_h(x)=
\int\limits_{x\in M}i_M^*\tau\big(A^\divideontimes AP_x\big)f(x)\mu_g(x)
\leqq F_{M,g}(A^\divideontimes A)\Vert f\Vert_{L^\infty(M)}\]
hence $I(M,g)\subseteqq I(M,h)$ and likewise we see that $I(M,h)\subseteqq 
I(M,g)$ {\it i.e.}, $I(M,h)=I(M,g)$. It readily follows from the structure of 
$F_{M,g}$ that 
$A\in I(M,g)$ if and only if $AP_x=0$ for $\mu_g$-a.e. $x\in M$ 
hence by continuity of (\ref{beagyazas}) for all $x\in M$. This demonstrates 
that, being $P_x$ surely not invertible, non-trivial solutions in principle 
are allowed and belong to the infinitesimal ideals 
$0\not=\fr(1-P_x)\subsetneqq\fr$ for all $x\in M$. These are ultraweakly 
closed left-ideals which in fact are, via diffeomorphisms of $M$, unitary 
equivalent. Therefore, considering the things so far merely as a motivation, 
for a connected compact oriented smooth $4$-manifold $M$ hereby we define 
\begin{equation}
I(M):=\bigcap\limits_{x\in M}\fr\:(1-P_x)
\label{ideal}
\end{equation}
which is by construction an ultraweakly closed left-ideal satisfying 
$0\subseteqq I(M)\subsetneqq\fr$ and is invariant under diffeomorphisms 
(acting as inner automorphisms of $\fr$).

Recall that the standard representation $\pi:\fr\rightarrow\fb(\ch)$ arises via 
multiplication from the left in $\fr $ on itself. Since 
$0\subseteqq I(M)\subsetneqq\fr$ is a left-ideal $\pi$ restricts to a 
representation of $\fr$ on the Hilbert space completion 
$0\subseteqq\ci(M)\subseteqq\ch$ of $I(M)$ given by (\ref{ideal}). 
The scalar product on $\fr\subset\ch$ looks like 
$(\hat{A},\hat{B})=\tau(AB^\divideontimes)$ hence satisfies 
$(\widehat{AB},\widehat{C}\:)=(\hat{B},\widehat{A^\divideontimes C})$; 
consequently the standard representation restricts to the orthogonal 
complementum $0\subseteqq\ci(M)^\perp\subseteqq\ch$ as well. Note that 
$\ci(M)^\perp$ as a complete complex vector is isomorphic to $\ch/\ci(M)$. 
Therefore in the spirit of the Gelfand--Naimark--Segal construction we want to 
take the representation of $\fr$ on $\ci(M)^\perp$ hence for a given 
$M$ we introduce the representation $\rho_M:\fr\rightarrow\fb (\ci(M)^\perp)$ 
to be simply the restricted representation $\pi\vert_{\ci(M)^\perp}$. 
The general theory, namely \cite[Proposition 8.2.3]{ana-pop} says that if 
$Q_M:\ch\rightarrow\ci (M)^\perp$ is the orthogonal projection then 
$Q_M\in\fr'$ because for the left $\fr$-module 
$0\subseteqq\ci(M)\subseteqq\ch$ holds {\it i.e.}, it lies in the standard 
one (and not in $\ch\otimes\ell^2(\N)$ as in general). The Murray--von 
Neumann coupling constant of $\rho_M$ is therefore equal to 
$\tau (Q_M)\in[0,1]$. Note that the coupling constant depends 
only on the unitary equivalence class of $\rho_M$ hence in particular it 
is preserved by orientation-preserving diffeomorphisms of $M$ (which are 
unitary on $\ch$). We conclude that 
$q(M):=\tau (Q_M)\in [0,1]$ is a smooth invariant of the compact 
$4$-manifold $M$.

Concerning an important restriction on the spectrum of $q$, let 
$\{U_i\}_{i=1,\dots,m}$ be a finite open covering of $M$ consisting of small 
open $4$-balls and for every $U_i\subset M$ let 
$U_i\Subset U'_i\subseteqq M$ be the union of those members of the open 
covering which overlap with $U_i$ as before. By the general theory of 
ultraweakly closed left-ideals in a von Neumann algebra 
\cite[Theorem 6.1.12]{jon}, there exist unique projections 
$P_{U'_i}\in\fr$ such that $\bigcap\limits_{x\in U'_i}\fr(1-P_x)=
\fr(1-P_{U'_i})$ for every $i=1,\dots,m$. Hence introducing 
$I(U'_i):=\fr(1-P_{U'_i})$ 
the uncountable infinite decomposition in (\ref{ideal}) can be coarse-grained 
to a finite one $I(M)=\bigcap\limits_{i=1}^mI(U'_i)$. Recall that the 
geometric meaning of the abstract projection $P_x\in\fr$ assigned to the point 
$x\in M$ is that it acts as the identity on $\Omega^2(M,x;\C)$ too. 
Taking an arbitrary operator $A\in\bigcap\limits_{x\in U'_i}\fr(1-P_x)
\cap C(M)$, then $AP_x=0$ implies $0=(AP_x)\Omega^2(M,x;\C)=
A(P_x\Omega^2(M,x;\C))=A\Omega^2(M,x;\C)$ for every $x\in U'_i$ or 
writing equivalently, $0=A\Omega^2(M,\overline{U'_i};\C)$. 
Hence $\Omega^2(M,\overline{U'_i};\C)\subseteqq\bigcap\limits_A\ker A$ moreover 
we can suppose that $\bigcap\limits_A\ker A\subseteqq
\Omega^2(M,\overline{U_i};\C)$ on the smaller subset $U_i\Subset U'_i$. 
However also $A\in\fr(1-P_{U'_i})\cap C(M)$ that is, $AP_{U'_i}=0$ which 
implies that $0=(AP_{U'_i})\Omega^2(M;\C)=A(P_{U'_i}\Omega^2(M;\C))$, dictating 
the abstract projection $P_{U'_i}\in\fr$ to act on $\Omega^2(M;\C)$ too, such 
that $P_{U'_i}\Omega^2(M;\C)\subseteqq\bigcap\limits_A\ker A\subseteqq
\Omega^2(M,\overline{U_i};\C)$. In this way we recognize by referring to 
the discussion above (\ref{beagyazas}) that $P_{U'_i}\in\fr$ is nothing but an 
approximating projection of an elementary one $P_{x_i}\in\fr$ such that 
$x_i\in U_i$. Consequently, if the open covering of $M$ is fine 
enough, we can suppose that the corresponding projections are already 
sufficiently close to each other {\it i.e.}, $[[P_{x_i}-P_{U'_i}]]<1$ for every 
$i=1,\dots,m$ where $[[\:\cdot\:]]$ is the norm on $\fr$. However in this 
situation the projections $P_{x_i},P_{U'_i}\in\fr$ are 
unitary equivalent (in general within $\fb(\ch)$ only), hence the same 
holds for their corresponding ideals $\fr(1-P_{x_i})\subsetneqq\fr$ and 
$\fr(1-P_{U'_i})=I(U'_i)\subsetneqq\fr$. We conclude that, in the 
case of a sufficiently fine but yet finite open covering $\{U_i\}_{i=1}^m$ of 
$M$, the already known fact $0\not=\fr(1-P_{x_i})$ implies that $0\not=I(U'_i)$ 
too, for every $x_i\in U_i\Subset U'_i$ and $i=1,\dots,m$. 

Consider $U_1\subset M$ and take any $0\not=A_1\in I(U'_1)$. Likewise, 
take $U_2\subset M$ and pick $0\not=A_2\in I(U'_2)$ such that 
$0\not=A_1A_2$ (we can assume this by slightly perturbing the factors if 
necessary). On the one hand, being $I(U'_1)\subsetneqq\fr$ a left-ideal, 
it is acted upon by $\fr$ from the left therefore its Hilbert space 
completion $\ci(U'_1)\subseteqq\ch$ is a left $\fr$-module. Take again 
by \cite[Proposition 8.2.3]{ana-pop} the unique projection 
$Q_{U'_1}\in\fr'$ satisfying $\ci(U'_1)=(1-Q_{U'_1})\ch$. 
Then there exists an $\xi\in\ch$ such that $\hat{A}_1=(1-Q_{U'_1})\xi$. But 
this means that there exists a sequence of elements in $\fr$ such that its 
corresponding sequence $\{\hat{B}_j\}_{j=1,2,\dots}$ converges to $\xi$ in the 
norm $\Vert\:\cdot\:\Vert$ on $\ch$ as $j\rightarrow+\infty$ hence by the 
continuity of multiplication we get 
$(1-Q_{U'_1})\hat{B}_j\rightarrow\hat{A}_1$. 
Hence $(1-Q_{U'_1})\widehat{B_jA}_2\rightarrow\widehat{A_1A}_2$ 
as $j\rightarrow+\infty$ yielding $\widehat{A_1A}_2\in\ci(U'_1)$. On 
the other hand $A_1A_2\in I(U'_2)$ implies $\widehat{A_1A}_2\in\ci(U'_2)$ too. 
Thus $0\not=\widehat{A_1A}_2\in\ci(U'_1)\cap\ci(U'_2)=
\ci(U'_1\cup U'_2)$. Repeating this procedure we come up with a finite product 
$0\not=\widehat{A_1A_2\dots A}_m\in\bigcap\limits_{i=1}^m\ci(U'_i)=\ci(M)$. 

Eventually we find $0\subsetneqq\ci(M)\subseteqq\ch$ hence 
for the corresponding projection $0<1-Q_M\leqq 1$ demonstrating that 
$0\leqq q(M)<1$. Therefore essentially the compactness of $M$ implies that 
$q(M)\in[0,1)$.
\end{proof}

\noindent Elements of Jones' subfactor theory (for a summary cf. {\it e.g.} 
\cite[Section 9.4]{ana-pop} or \cite[Chapter V.10]{con} or 
\cite[Chapter 19]{jon}) offer a computational tool for $q$ just introduced.

\begin{lemma} Let $M$ be a connected compact oriented smooth $4$-manifold and
$q(M)\in[0,1)$ its smooth invariant as before. Then there exists a 
subfactor $0\not=\fii(M)\subsetneqq\fr$ satisfying 
$4\leqq[\fr:\fii(M)]<+\infty$ such that 
$q(M)=1-\frac{4}{[\fr\::\:\fii(M)]}$\:\:. 
\label{jones}
\end{lemma}

\begin{proof} We continue working with the left-ideal (\ref{ideal}) 
distilled out of the smooth structure (and the orientation) of $M$ alone. 
Put
\begin{equation}
\fii(M):=I(M)^\divideontimes I(M)\:\:.
\label{idef}
\end{equation}
We assert that $\fii(M)$ is the largest von Neumann subalgebra within 
$I(M)$ and a non-trivial subfactor within $\fr$. (Meanwhile observe that 
$I(M)I(M)^\divideontimes=\fr$ always.) Indeed, by the general theory 
\cite[Theorem 6.1.12]{jon}, since $0\not=I(M)\subsetneqq\fr$, there 
exists a unique projection in $\fr$ satisfying $0\not=P_M\not=1$ 
such that $I(M)=\fr(1-P_M)$ holds.\footnote{We are using a convention 
compatible with (\ref{ideal}). However note that there is a slight 
abuse of notation here because we also introduced $M'$ and a 
corresponding projection $P_{M'}$ before. Although $M'=M$ note that 
$P_{M'}\not=P_M$ because we put $P_{M'}=0$ before but surely $P_M\not=0$ here. 
Fortunately this will not cause any confusion in what follows.} 
For $\fr$ is a factor hence multiplication on it is surjective, the equalities 
$I(M)^\divideontimes\cap I(M)=(1-P_M)\fr(1-P_M)=I(M)^\divideontimes 
I(M)$ follow at once. However by \cite[Corollary 3.4.4]{jon} we know that 
$(1-P_M)\fr(1-P_M)\subset\fr$ is a non-trivial subfactor (abstractly 
isomorphic to $\fr$), hence the assertion also follows. 

Let us compute the index $[\fr:\fii(M)]
\in\big\{4\cos^2\big(\frac{\pi}{n}\big)\:\big\vert\:n=3,4,\dots\big\}
\cup[4,+\infty]$. First of all, observe that both $0\not=P_M\not=1$ 
and $0\not=1-P_M\not=1$ belong to $\fr\cap\fii(M)'$ hence this intersection 
cannot be trivial; however it is known \cite[Corollary 19.2.4]{jon} that 
$\fr\cap\fii(M)'=\C1$ whenever $1\leqq [\fr:\fii(M)]<4$ 
consequently the index spectrum of 
(\ref{idef}) is restricted to $[\fr:\fii(M)]\in[4,+\infty]$. 
Proceeding further, via Footnote \ref{dimenzio} we can write 
that $q(M)=\tau(Q_M)=\dim_\fr\ci(M)^\perp$ and the standard representation 
satisfies $\dim_\fr\ch=1$ hence $\dim_\fr\ci(M)=1-q(M)$. Moreover, 
restricting the action of $\fr$ on $\ci(M)$ to its subalgebra $\fii(M)$ we 
can render $\ci(M)$ a left $\fii(M)$-module too. Concerning its dimension,
recall the general fact \cite[Appendix A.2]{ana-pop} that in a
von Neumann algebra any element can be written as a $\C$-linear combination
of at most $4$ positive elements {\it i.e.}, elements of the form
$B^\divideontimes B$ with $B\in\fr$. Taking then $A\in I(M)$ and
$P_x\in\fr$ from (\ref{beagyazas}) and inserting
$A=b_1B_1^\divideontimes B_1+\dots+b_4B_4^\divideontimes B_4$ into $AP_x=0$
and leveraging the linear independence of the terms,
we get $B_i^\divideontimes B_iP_x=0$ hence
$0=\big(\widehat{B_i^\divideontimes B}_i,\hat{P}_x\big)=
\Vert\widehat{B_iP_x}\Vert^2$ yielding $B_iP_x=0$ for all $x\in M$
that is, $B_i\in I(M)$
consequently $B^\divideontimes_iB_i\in I(M)^\divideontimes I(M)=\fii(M)$ for
all $i=1,2,3,4$. Considering the standard representation of $\fii(M)$ on its own
Hilbert space completion having ${\fii(M)}$-dimension $1$, we conclude that
$\ci(M)$ as a left $\fii(M)$-module decomposes into $4$ copies of the
standard representation yielding $\dim_{\fii(M)}\ci(M)=4$. Thus by
\cite[Proposition 19.2.3]{jon} we compute that
$[\fr:\fii(M)]=\frac{\dim_{\fii(M)}\ci(M)}{\dim_\fr\ci(M)}=\frac{4}{1-q(M)}$.
Finally, since $q(M)\not=1$ via Lemma \ref{terbelilemma}, we are sure that
$[\fr:\fii(M)]\not=+\infty$ hence the result.
\end{proof}

\noindent Next we collect some basic useful properties of the invariant.

\begin{lemma} (Reversing orientation.) If $M$ is a connected compact 
oriented smooth $4$-manifold and $\overline{M}$ is its orientation-reversed 
form then $q(\overline{M})=q(M)$.

(Gluing principle.)  Let $M$ and $N$ be two connected compact oriented 
smooth $4$-manifolds and let $M\#N$ be their connected sum. With induced
orientation $M\#N$ is a connected compact oriented smooth $4$-manifold. 
Then   
\[q(M\#N)=q(M)+q(N)-q(M)q(N)\]
and in particular if $\s^n\subset\R^{n+1}$ is the standard $n$-sphere then 
$q(\s^4)=0$.

(Blow-up.) If $M'$ is a smooth blow-up of $M$ then
\[q(M)\leqq q(M')=q(M)+\big(1-q(M)\big)t\]
where $0\leqq t<1$ is given by $t:=q(\C P^2)$.
\label{gammaformulak}
\end{lemma}

\begin{proof} The first assertion is obvious from
$q$'s construction carried out in the proof of Lemma \ref{terbelilemma}.

Concerning the second assertion, let $\sm$ denote the category of
orientation-preserving diffeomorphism classes of all connected closed
oriented smooth $4$-manifolds. As a general observation, the
$q$-invariant is a well-defined map from the set of objects of $\sm$
into $[0,1)\subset\R$. But these comprise
a commutative semigroup with (one possible) unit $\s^4$ under the
connected sum operation $\#$. That is, if $X,Y,Z,\s^4\in\sm$
then $X\#Y\cong Y\#X$ and $(X\#Y)\#Z\cong X\#(Y\#Z)$ and
$X\#\s^4\cong X$. Pick $M,N,M\#N\in\sm$
and consider the corresponding $q(M),q(N),q(M\#N)\in [0,1)$. Introduce
$\bullet :[0,1)\times [0,1)\rightarrow [0,1)$
by setting $q(M)\bullet q(N):=q(M\#N)$. The $\bullet$-operation is
therefore well-defined having the properties $q(X)\bullet q(Y)=q(Y)\bullet
q(X)$ and $(q(X)\bullet q(Y))\bullet q(Z)=q(X)\bullet (q(Y)\bullet q(Z))$
moreover satisfying $q(X)\bullet q(\s^4)= q(X)$ {\it i.e.}, $q(\s^4)$ is a
unit. These ensure us that $([0,1),\bullet )$ is a unital commutative
semigroup and $q:(\sm,\#)\rightarrow ([0,1),\bullet )$ is a unital semigroup
homomorphism. As a specific observation note that $q$ has been defined in
Lemma \ref{terbelilemma} as the (continuous) dimension of a closed subspace
within the $\fr$-module Hilbert space $\ch$; more precisely
$q(M)=\tau(Q_M)$ arises through a chain of assignments
$M\mapsto\ci(M)^\perp\mapsto\dim_\fr\ci(M)^\perp$. Lemma \ref{terbelilemma}
and in particular (\ref{ideal}) imply that $\ci(M\#N)=\ci(M)\cap\ci(N)$
for the Hilbert space completion. Introducing the Abelian
group $(L(\ch),+)$ of closed subspaces with respect to taking
sum and hence $0\in\ch$ playing the role of the unit we see
that the assignment $M\mapsto\ci(M)^\perp$ has the property
$\ci(M\#N)^\perp=(\ci(M)\cap\ci(N))^\perp=
\ci(M)^\perp+\ci(N)^\perp$ hence induces a unital semigroup
homomorphism $L:(\sm,\#)\rightarrow(L(\ch)\setminus\{\ch\},+)$. 
Thus $q$ factorizes as 
\[q:(\sm,\#)\xrightarrow{\hspace{0.1in}L\hspace{0.1in}}
(L(\ch)\setminus\{\ch\},+)\xrightarrow{\hspace{0.1in}
\dim_\fr\hspace{0.1in}}([0,1),\bullet)\]
consequently the unique plain dimension function on subspaces as it is, must 
in fact be an Abelian unital semigroup homomorphism too hence its properties 
impose constraints on its target structure. Recall the $q$-invariant when 
evaluated on a connected sum has been written as 
$q(M\#N)=q(M)\bullet q(N)$ expressing that it depends only 
on $q(M)$ and $q(N)$. Putting $\cv=\ci(M)^\perp$ and 
${\mathscr W}=\ci(N)^\perp$ and knowing that the dimension function behaves 
like $\dim_\R(\cv+{\mathscr W})=\dim_\fr\cv+\dim_\fr{\mathscr W}-
\dim_\fr(\cv\cap{\mathscr W})$ forces to write 
$q(M)\bullet q(N)=q(M)+q(N)-f\big(q(M),q(N)\big)$ 
and the unknown function is strongly determined by the geometric properties of 
intersection of subspaces. Namely we know that 
$f:[0,1)\times[0,1)\rightarrow[0,1)$ 
such that $f(s,t)=f(t,s)$ and $s\mapsto f(s,t)$ is linear, $f(0,t)=0$ 
and $\lim_{s\rightarrow 1}f(s,t)=t$. The function 
$f(s,t)=st$ solves these constraints. This forces the Abelian semigroup 
multiplication law on $[0,1)$ to look like 
$s\bullet t=s+t-st$ with $0\in[0,1)$ being the unit hence yielding the shape 
for $q(M\#N)$ moreover constraining $q(\s^4)=0$ as stated.

Finally, since the blow-up in the smooth category $(\sm,\#)$ is by definition 
given by $M':=M\#\overline{\C P^2}$ the result follows if we put 
$t:=q(\overline{\C P^2})=q(\C P^2)$ in the connected sum formula. 
\end{proof}

\noindent{\it Proof of Theorem \ref{fotetel2}}. Putting together 
Lemmata \ref{terbelilemma}, \ref{jones} and \ref{gammaformulak} the theorem 
follows.\hspace{2.5cm}$\square$
\vspace{0.05in}

\noindent {\it An excursus on the invariant $q$}. Before proceeding 
further let us take a closer look on the introduced smooth $4$-manifold 
invariant concerning its effective computability. The general experience 
is that the more sensitive an invariant is, the less computable it is. 
By techniques borrowed from $4$-manifold theory and the gluing principle above, 
the following non-trivial but quite insensitive behaviour of 
$q$ shows up in the finite fundamental group situation.

\begin{lemma}
Let $M'$ and $M''$ be closed connected orientable smooth $4$-manifolds 
having finite fundamental groups. If $M'$ and $M''$ are homeomorphic then 
$q(M')=q(M'')$. 

In fact if $M$ is as above then 
\[q(M)=1-\vert\pi_1(M)\vert(1-t)^{{\rm rk}(\pi_2(M))}\]
with $t=q(\C P^2)$ as before and the condition $0\leqq q(M)<1$ makes sure 
that $0<t<1$. 
\label{egyszeresen-osszefuggo}
\end{lemma}

\begin{proof} Assume first that both $M'$ and $M''$ are closed, 
connected, simply connected (hence orientable) and are homeomorphic. 
Then by {\it e.g.} \cite[Theorem 9.1.12]{gom-sti} there exists an 
integer $k\geqq 0$ such that $M'\#k(\C P^1\times\C P^1)\cong M''\#k(\C 
P^1\times\C P^1)$ yielding $q(M'\#k(\C P^1\times\C P^1))=q(M''\#k(\C 
P^1\times\C P^1))$. Introducing $s:=q(k(\C P^1\times\C P^1))$ and 
applying the gluing principle from Lemma \ref{gammaformulak} 
we obtain an equality $q(M')+s-q(M')s=q(M'')+s-q(M'')s$. Finally 
leveraging the invertability of the map $r\mapsto r+s-rs$ when $s<1$, we 
come up with $q(M')=q(M'')$.

Concerning the explicit formula in the simply connected realm first, 
with some $s\in [0,1)$ take the recursive sequence 
\[R_0(s):=0\:,\:R_1(s):=s\:,\:\dots\:,\: R_k(s):=s+R_{k-1}(s)-sR_{k-1}(s)
\:,\:\dots\] 
for all $k=0,1,2,\dots$ representing the semigroup $\{0\}\cup\N$ inside 
$[0,1)$ and put $t:=q(\C P^2)=q(\overline{\C P^2})$. Now 
if $M_1$ and $M_2$ are connected,
closed, simply connected, smooth then there exist integers $k_1,l_1\geqq 0$ and
$k_2,l_2\geqq 0$ such that $M_1\#k_1\C P^2\#l_1\overline{\C P^2}\cong
M_2\#k_2\C P^2\#l_2\overline{\C P^2}$ (cf. {\it e.g.} 
\cite[Theorem 9.1.14]{gom-sti}) which gives again that
$q(M_1\#k_1\C P^2\#l_1\overline{\C P^2})=
q(M_2\#k_2\C P^2\#l_2\overline{\C P^2})$. Then by the gluing principle
\[q(M_1)+R_{k_1+l_1}(t)-q(M_1)R_{k_1+l_1}(t)=
q(M_2)+R_{k_2+l_2}(t)-q(M_2)R_{k_2+l_2}(t)\:\:.\]
Let $M_1:=M$ be arbitrary and $M_2:=\s^4$ hence $q(M_2)=0$.
Then we can suppose that $k_1+l_1\leqq k_2+l_2$ therefore 
$q(M)+R_{k_1+l_1}(t)-q(M)R_{k_1+l_1}(t)=R_{k_2+l_2}(t)$ 
from which again, since $t<1$, by invertability we find
$q(M)=R_{k_2+l_2-k_1-l_1}(t)$ hence setting 
$n:=k_2+l_2-k_1-l_1\geqq0$ we get $q(M)=R_n(t)$. It is 
clear from the proof that here $n=b_2(M)$ i.e., $n={\rm rk}(\pi_2(M))$ by 
Hurewicz's theorem, moreover it is easy to see that 
$R_k(s)=1-(1-s)^k$ is the solution of the above recursion. 
Hence inserting $k:={\rm rk}(\pi_2(M))$ and 
$s:=t$ and knowing that $\vert\pi_1(M)\vert=1$ we end up with the stated 
formula for $q(M)$ in the simply connected situation.

Concerning the non-simply connected case, first recall the well-known 
fact that every finite group $G$ appears as the fundamental group of some 
closed, connected and orientable smooth $4$-manifold $M$ i.e., 
$\pi_1(M)\cong G$. Consider the universal covering $\widetilde{M}$ which 
is therefore remains closed, and is connected and simply connected. Then 
$\widetilde{M}/G\cong M$ consequently $\widetilde{M}$ is acted upon 
freely by $G$ and $M$ can be identified with an open subset of the 
universal covering space, namely one particular fundamental domain 
$F(M)\subset\widetilde{M}$ of this $G$-action on $\widetilde{M}$. First: 
this inclusion, by referring to (\ref{ideal}), implies that 
$I(\widetilde{M})\subset I(F(M))$ hence via (\ref{idef}) we also obtain an 
inclusion $\fii(\widetilde{M})\subset\fii(F(M))$ of von Neumann 
algebras. Second: on the one hand as $G$ permutes the fundamental 
domains, it acts via outer $\divideontimes$-automorphisms on 
$\fii(F(M))$ exhibiting a fixed-point subalgebra 
$\fii(F(M))^G\subset\fii (F(M))$; on the other hand the obvious 
inclusion $G\subset{\rm Diff}^+(\widetilde{M})$, compared with 
(\ref{ideal}) which says that $I(\widetilde{M})$ is an intersection of 
pointwise ideals permuted by ${\rm Diff}^+(\widetilde{M})$ hence 
$I(\widetilde{M})$ is pointwisely fixed by ${\rm Diff}^+(\widetilde{M})$ thus 
by (\ref{idef}) so is $\fii(\widetilde{M})$, 
gives $\fii(\widetilde{M})=\fii(F(M))^G$. Third: 
the existence of an orientation-preserving diffeomorphism between $F(M)$ 
and an open subset of $M$ having measure-zero complementum hence an 
isomorphism $I(F(M))\cong I(M)$, implies an isomorphism 
$\fii(F(M))\cong\fii(M)$ of von Neumann algebras and we already know 
that the latter is a subfactor of $\fr$ hence so is $\fii(F(M))$. Now 
putting together these three pieces of data and referring to \cite[Corollary 
19.1.6]{jon} we observe that $\fii(\widetilde{M})\subset\fii(F(M))$ is 
an inclusion of subfactors having index 
$[\fii(F(M)):\fii(\widetilde{M})]=\vert G\vert$. Consequently, by 
referring to Lemma \ref{jones}, we compute 
$q(\widetilde{M})=1-\frac{4}{[\fr:\fii(\widetilde{M})]}=1-\frac{4}{[\fr: 
\fii(F(M))][\fii(F(M)):\fii(\widetilde{M})]}=1-\frac{4}{[\fr:\fii(M)] 
\vert G\vert}$ and then multiply both sides with $1<\vert 
G\vert<+\infty$ and insert $q(M)=1-\frac{4}{[\fr:\fii(M)]}$ too. This way we 
find that $q(M)=1+\vert G\vert(q(\widetilde{M})-1)$. Finally plugging 
$G\cong\pi_1(M)$ as well as the already known expression 
$q(\widetilde{M})=1-(1-t)^{{\rm rk}(\pi_2(\widetilde{M}))}$ from the simply 
connected case and recalling that covering implies an 
isomorphism $\pi_2(\widetilde{M})\cong\pi_2(M)$ we come up with the 
desired formula in the non-simply connected case too. 
\end{proof}

\noindent As an important task next we compute 
the number $0<t<1$ apprearing in both Lemmata \ref{gammaformulak} and 
\ref{egyszeresen-osszefuggo} and is assigned to $\C P^2$. 

\begin{lemma} An equality $q(\C P^2)=\frac{5}{9}$ holds hence 
$t=\frac{5}{9}$. Thus $q(M)=1-\vert\pi_1(M)\vert\big(\frac{2}{3} 
\big)^{2\:{\rm rk}(\pi_2(M))}$ in Lemma \ref{egyszeresen-osszefuggo}. Moreover 
the condition $0\leqq q(M)<1$ gives rise to an inequality 
$\vert\pi_1(M)\vert\leqq\big(\frac{3}{2}\big)^{2\:{\rm rk}(\pi_2(M))}$ 
whenever $M$ is a closed, connected, orientable smooth $4$-manifold having 
finite fundamental group.
\label{tlemma}
\end{lemma}

\begin{proof} We have seen that the action 
$x\mapsto\Phi(x)$ of a diffeomorphism on a point $x\in M$ 
induces the unitary action $P_x\mapsto(\Phi^{-1})^*P_x\Phi^*$ on the 
corresponding projection $P_x\in\fr$ by (\ref{beagyazas}) hence this action is 
an inner $\divideontimes$-automorphism of $\fr$. In this way we obtain a 
commutative diagram
\[\xymatrix{\fr\ar[r]^{{\rm Ad}_{(\Phi^{-1})^*}} &
              \fr\\
          M\ar[u]^{i_M} \ar[r]^{\Phi}  & M\ar[u]^{i_M}}\]
which governs the inclusion of ${\rm Diff}^+(M)$ as a subgroup of inner 
$\divideontimes$-auto\-morphisms into ${\rm Aut}\:\fr$. It is clear that the 
corresponding subfactor $\fii(M)$ from Lemma \ref{jones} must be invariant 
within $\fr$ under diffeomorphisms {\it i.e.}, is preserved by 
${\rm Diff}^+(M)$ acting as inner $\divideontimes$-automorphisms on $\fr$. 
Assume furthermore that there exists a {\it compact} subgroup 
$G\subsetneqq{\rm Diff}^+(M)$ which acts transitively on $M$. Thus taking into 
account the decomposition in 
(\ref{ideal}) we can see by picking any point $x_0\in M$ that 
\[I(M)=\bigcap\limits_{x\in M}\fr(1-P_x)= \bigcap\limits_{\Phi\in 
{\rm Diff}^+(M)}\fr\big(1-P_{\Phi(x_0)}\big)= \bigcap\limits_{\Psi\in 
G}\fr\big(1-P_{\Psi(x_0)}\big)\] 
that is, the ideal already arises as taking intersection under the subgroup 
only. Moreover being an intersection, we have already observed that $I(M)$ in 
fact is pointwisely fixed by $G$ 
hence by (\ref{idef}) so is $\fii(M)$. Therefore an effective form of 
the invariance condition is that $\fii(M)$ belongs to the 
{\it fixed point subalgebra} $\fr^G$ of a {\it compact} group $G$ 
operating as inner $\divideontimes$-automorphisms on $\fr$, dictated by 
the diagram above.

Consider now the Fubini--Study metric $g$ on the complex projective 
plane $\C P^2$. It is a classical fact that this is a 
homogeneous Riemannian $4$-manifold satisfying the Einstein condition with 
a non-zero cosmological constant. 
The isometry group ${\rm Iso}^+(\C P^2,g)\subsetneqq{\rm Diff}^+(\C P^2)$ of 
the Fubini--Study metric is large enough to act transitively along $\C P^2$. 
A way to implement ${\rm Iso}^+(\C P^2,g)\cong{\rm PU}(3)$ into 
$\fb(\ch)$ to meet our demands is to take conjugations by $3\times 3$ 
unitary matrices on $\fr\cong\fm_3(\fr)\cong\fm_3(\fii(\C P^2))$. Consequently 
$[\fr:\fii(\C P^2)]=[\fm_3(\fii(\C P^2)):\fii(\C P^2)]=3^2=9$ demonstrating 
by Lemma \ref{jones} that 
$q(\C P^2)=1-\frac{4}{[\fr\::\:\fii(\C P^2)]}=1-\frac{4}{9}=\frac{5}{9}=t$ 
as stated.  
\end{proof}

%Concerning $\s^1\times\s^3$ we repeat the previous steps. Choosing an 
%isomorphism $\s^1\times\s^3\cong{\rm U}(2)$ and taking the corresponding 
%biinvariant metric $g$ we can make $\s^1\times\s^3$ 
%a homogeneous Riemannian $4$-space, acted upon transitively 
%by ${\rm PU}(2)={\rm PU}(2)\times 1\subset
%{\rm U}(2)\times{\rm U}(2)/Z({\rm U}(2))={\rm Iso}^+({\rm U}(2),g)$. Replacing 
%${\rm PU}(3)$ with ${\rm PU}(2)$ 
%we get $[\fr:\fii(\s^1\times\s^3)]=[\fm_2(\fii(\s^1\times\s^3)):
%\fii(\s^1\times\s^3)]=2^2=4$ 
%consequently $q(\s^1\times \s^3)=1-\frac{4}{[\fr\::\:\fii(\s^1\times\s^3)]}
%=1-\frac{4}{4}=0$ as stated. 
%\end{proof}

\noindent We close our detour on $q$ as well as this section
with two general observations. The first is as follows: 
Lemma \ref{egyszeresen-osszefuggo} unfortunately makes sure that in the 
finite-fundamental-group regime at least, $q$ is not really sensitive 
because it depends only on two simple topological data. In fact 
$q$ gets less-and-less injective as the sizes of the $1^{\rm st}$ and 
the $2^{\rm nd}$ homotopy groups of a topological $4$-manifold carrying smooth 
structure(s) increases. On the contrary, 
as these parameters approach zero, $q$ has a sharp content: an application 
to $\s^4$ satisfying $q(\s^4)=0$ implies that if the smooth four dimensional 
Poincar\'e conjecture fails then $q$ is not injective at zero too. 
A stronger assertion is the following.

\begin{lemma} 
Consider $q$ as a restricted map from (the objects of)
the category ${\sm_{fin}}$ of isomorphism classes of all connected,
closed, orientable, smooth $4$-manifolds having finite fundamental group,
into $[0,1)\subset\R$. Then $q$ is injective at
zero if and only if the smooth Poincar\'e conjecture holds {\it i.e.}, the
$4$-sphere admits a unique smooth structure.

\label{poincare1} 
\end{lemma}

\begin{proof} If $S^4$ is a differentiable manifold 
homeomorphic to $\s^4$ then $S^4\in\sm_{fin}$ such that $\pi_1(S^4)=1$ and 
$\pi_2(S^4)=0$ therefore Lemma \ref{egyszeresen-osszefuggo} implies 
$q(S^4)=1-1\cdot(1-t)^0=0$. Assume 
that $q\vert_{\sm_{fin}}$ is injective at zero. Then $S^4\cong\s^4$ and the 4 
dimensional smooth Poincar\'e conjecture follows.

Conversely, take $S^4\in\sm_{fin}$ such that $q(S^4)=0$. Then by Lemmata 
\ref{egyszeresen-osszefuggo} and \ref{tlemma} we conclude that 
$q(S^4)=1-\vert\pi_1(S^4)\vert\big(\frac{4}{9}\big)^{{\rm rk}(\pi_2(S^4))}=0$ 
hence necessarily $\vert\pi_1(S^4)\vert=1$ {\it i.e.}, $\pi_1(S^4)=1$ and 
${\rm rk}(\pi_2(S^4))=0$. Hurewicz's theorem and the 
universal coefficient formula together imply that $H_2(S^4;\Z)=0$ {\it i.e.}, 
$S^4$ has trivial intersection form. The 
validity of the 4 dimensional topological Poincar\'e conjecture 
\cite{fre} thus makes sure that $S^4$ is homeomorphic to $\s^4$. Assume that 
the smooth 4 dimensional Poincar\'e conjecture is true. Then $S^4\cong\s^4$ 
hence $q\vert_{\sm_{fin}}$ is injective at zero.
\end{proof} 

\noindent The content of Lemma \ref{poincare1} {\it i.e.}, the
interpretation of the preimage $q^{-1}(0)\subset\sm_{fin}$ can be
understood geometrically as the description of the ``fiber'' of $q$
considered as a map from $\sm_{fin}$ into $[0,1)\subset\R$. This strongly
motivates to move from understanding $q$ in individual cases $M\in\sm_{fin}$
towards its global behaviour over the whole (objects of) $\sm_{fin}$.
Our second observation in this direction starts with the following:

\begin{lemma} Assume that the image of $q$ as a restricted map from 
$\sm_{fin}$ as above into the semi-open real interval $[0,1)$ is everywhere 
dense. Then $q$ is nowhere injective and in fact the cardinalities of its 
preimage sets are always infinite. 
\label{poincare2} 
\end{lemma}

\begin{proof} For $r\in[0,1)$ let
$\vert q^{-1}(r)\vert=0,1,2,\dots,+\infty$ denote the cardinality of the
preimage set of $q$ at the point $r$ {\it i.e.}, the collection
of all pairwise non-diffeomorphic $4$-manifolds having $q$-invariants equal
to $r$. Here $\vert q^{-1}(r)\vert=0$ means $q^{-1}(r)=\emptyset$ {\it i.e.},
$r$ is not in the range of $q$; while the other extreme case
$\vert q^{-1}(r)\vert=+\infty$ means there exist infinitely many
pairwise non-diffeomorphic smooth $4$-manifolds, perhaps homeomorphic or not,
having $q$-invariants all equal to $r$. Consider the subset
$S:=\{s\in[0,1)\:\vert\:\vert q^{-1}(s)\vert>1\}
\subseteqq{\rm im}\:q\subseteqq[0,1)$. We assert that $S={\rm im}\:q$. The
proof is based on a modified open-closed argument adapted to the possibility
that ${\rm im}\:q$ is totally disconnected in $[0,1)$.

We observe first that $S$ is not empty. To be explicit let $M$ be a fixed $K3$
surface; it is known for a long time \cite[Corollary 3.3.7]{gom-sti}
that the simply connected topological $4$-manifold underlying $K3$ carries
countably infinitely many different smooth structures. Consequently
referring to Lemma \ref{egyszeresen-osszefuggo} and
putting $s:=1-\big(\frac{4}{9}\big)^{b_2(K3)}=1-\big(\frac{4}{9}\big)^{22}$ 
we know that $\vert q^{-1}(s)\vert=+\infty$ yielding $s\in S$.

Next we prove that $S$ is open in the following sense: for every inner
point $s\in (0,1)\cap S$ there exists $\varepsilon >0$ such that $t\in S$
for every $t\in(s-\varepsilon, s+\varepsilon)\cap{\rm im}\:q$
moreover $s\pm\varepsilon\in{\rm im}\:q$. First: consider the projection
$P_M\in\fr$ having been defined by $I(M)=\fr(1-P_M)$ in the proof of
Lemma \ref{jones}. There is a straightforward equality
$(1-P_M)\fr(1-P_M)=(1-P_M)\fii(M)(1-P_M)$ thus by \cite[Proposition
19.2.3]{jon} again, we compute
$1=[(1-P_M)\fr(1-P_M):(1-P_M)\fii(M)(1-P_M)]=\tau(1-P_M)T'(1-P_M)
[\fr:\fii(M)]$ where $T'$ is the normalized
trace on $\fii(M)'$. But $T'(1-P_M)=T(J(1-P_M)J)=T(1-P_M)=1$ since
$1-P_M$ obviously acts as the unit on $\fii(M)=(1-P_M)\fr(1-P_M)$.
Inserting this and the formula for $q(M)$ in Lemma \ref{jones} we get
$q(M)=4\tau(P_M)-3$. Second: put the ultraweak topology on $\fr$. On
the one hand the trace $\tau:\fr\rightarrow\C$ is then continuous.
On the other hand let us topologize (the set of objects in) $\sm_{fin}$ by
pulling back this topology from the dense subspace of projections of
$\fr$ by the aid of the map $M\mapsto P_M$. This map is therefore continuous
by construction. Hence $M\mapsto P_M\mapsto 4\tau(P_M)-3=q(M)$ is also a
continuous function from $\sm_{fin}$ into $[0,1)$.
Therefore the assignment $M\mapsto q(M)$ looks
like a branched covering over ${\rm im}\:q$ consequently
$t\mapsto\vert q^{-1}(t)\vert$ is
lower semicontinuous: for every inner point $s\in(0,1)\cap{\rm im}\:q$ there
exists $\varepsilon>0$ such that $\vert q^{-1}(t)\vert\geqq\vert
q^{-1}(s)\vert$ whenever $t\in(s-\varepsilon,s+\varepsilon)\cap{\rm
im}\:q$. Now applying this to any $s\in (0,1)\cap S$ we get $t\in S$ for
every $t\in(s-\varepsilon,s+\varepsilon)\cap{\rm im}\:q$ and notice that
such $t$'s always exist since by assumption ${\rm im}\:q$ is everywhere
dense within $[0,1)$. Moreover passing to a smaller $\varepsilon>0$ if
necessary, we can suppose in addition that the endpoints satisfy
$s\pm\varepsilon\in{\rm im}\:q$. Hence $S$ is open within ${\rm im}\:q$
and note that the complementum of this type of open interval is closed
within ${\rm im}\:q$.

Finally we prove closedness in the usual sense:
for every sequence $\{s_i\}_{i=1,2,\dots}$ within $S$ the
convergence to a point $t\in{\rm im}\:q$ implies $t\in S$. First: take a
fixed object $M\in\sm_{fin}$. It follows from
the strong$^\divideontimes$-limit-decom\-position (\ref{automorfizmus}) of an
automorphism of $\fr$ with respect to $M$, together with the
strong$^\divideontimes$-closedness of the embedding (\ref{beagyazas}) of
$M$ into $\fr$, that $M$ can be identified with the orbit of any of its
elementary projection $P_x\in\fr$ under ${\rm Aut}\:\fr$. Consequently if
$M,N\in\sm_{fin}$ are not diffeomorphic then their corresponding embeddings
(\ref{beagyazas}) are disjoint within $\fr$. In particular if
$x\in M$ and $y\in N$ then their $P_x,P_y\in\fr$ via (\ref{beagyazas})
are not unitary equivalent hence so are their ideals
$I(M),I(N)\subset\fr$ given by taking intersections as in (\ref{ideal});
thus their corresponding projections $1-P_M,1-P_N\in\fr$ as above,
are also not unitary equivalent. Therefore
$[[(1-P_M)-(1-P_N)]]=[[P_M-P_N]]=1$ where $[[\:\cdot\:]]$ is the norm on $\fr$.
Consequently at least one of the pairs $(1-P_M,1-P_N)$ or $(P_M,P_N)$
contains non-trivial orthogonal projections. Since the sum of the traces of
orthogonal projections cannot exceed $1\in[0,1]$ then from
$0\leqq q(M)=4\tau(P_M)-3<1$ {\it i.e.}, $\frac{3}{4}\leqq\tau (P_M)<1$
and likewise for $N$, we infer that $1-P_M,1-P_N$ are orthogonal yielding
$(1-P_M)(1-P_N)=0=(1-P_N)(1-P_M)$. Second: let $\{s_i\}_{i=1,2,\dots}$ be
a sequence in $S$ converging to a point
$t\in{\rm im}\:q$. Then for every $i=1,2,\dots$ we know
that $\vert q^{-1}(s_i)\vert>1$ therefore there exist $M_i,N_i\in\sm_{fin}$
satisfying $q(M_i)=s_i=q(N_i)$ such that $M_i\not\cong N_i$ in $\sm_{fin}$
hence $(1-P_{M_i})(1-P_{N_i})=0=(1-P_{N_i})(1-P_{M_i})$ in $\fr$.
Taking the limit $i\rightarrow+\infty$
by continuity of $q$ there exist $M,N\in\sm_{fin}$ such that $q(M)=t=q(N)$.
However the corresponding projections are non-trivial
and by the constructed ultraweak continuity of the assignment $M\mapsto P_M$
(in which multiplication is continuous) they yet satisfy $(1-P_M)(1-P_N)=0=
(1-P_N)(1-P_M)$. Therefore $P_M\not=P_N$ yielding $M\not\cong N$.
Consequently $\vert q^{-1}(t)\vert>1$ continues to hold {\it i.e.}, $t\in S$.
Thus $S$ is closed within ${\rm im}\:q$ as well.

Thus $S$ is not empty, $S$ is open in ${\rm im}\:q$ with closed complementum
morever $S$ is closed in ${\rm im}\:q$ as well hence our conclusion
is therefore that $S={\rm im}\:q$ and the assertion follows.
We also conclude that the cardinality of $q^{-1}(t)$ for every
$t\in{\rm im}\:q\subseteqq[0,1)$ is in fact infinite.
\end{proof}

\noindent Therefore our second observation above concerning the global 
properties of $q$ raises the important question whether or not the image 
of $\sm_{fin}$ under $q$ is dense within the real semi-open interval 
$[0,1)$. Concerning this point we can make only a quite straightforward 
further step here as follows.

\begin{lemma} Consider $\sm_{fin}$ as in Lemma
\ref{poincare2} and suppose that its set of objects possesses at least one
of the two properties below:
\begin{itemize}

\item[(i)] There is a sequence $\{M_m\}_{m=1,2,\dots}\subset\sm_{fin}$
of $4$-manifolds consisting of only finitely many homeomorphic copies of 
the $4$-sphere and the associated sequence of their homotopy
groups satisfying $\left\vert\frac{\log\vert\pi_1(M_m)\vert}
{\log\frac{9}{4}}- {\rm rk}(\pi_2(M_m))\right\vert\rightarrow 0$ as 
$m\rightarrow+\infty$; or putting equivalently: 
$\vert\pi_1(M_m)\vert\rightarrow+\infty$ and
$\left\vert\frac{\log\vert\pi_1(M_m)\vert}{\log\frac{9}{4}}- {\rm rk}
(\pi_2(M_m))\right\vert\rightarrow 0$ as $m\rightarrow+\infty$;

\item[(ii)] Take an inner point $r\in(0,1)$ and an associated sequence
$\{(m_n,n)\}_{n=1,2,\dots}$ of pairs of integers with $m_n>0$
satisfying $(m_n-1)(\frac{2}{3})^{2n}\leqq1-r\leqq m_n(\frac{2}{3})^{2n}$
or $m_n(\frac{2}{3})^{2n}\leqq 1-r\leqq (m_n+1)(\frac{2}{3})^{2n}$.
Then for every real number $r\in(0,1)$ and pair of positive integers
$(m_n,n)$ from its approximating sequence there exists an $M\in\sm_{fin}$
such that $\big(\vert\pi_1(M)\vert, {\rm rk}(\pi_2(M))\big)=(m_n,n)$.
\end{itemize}
In the validity of either cases it follows that ${\rm im}\:q$ is dense in 
$[0,1)$.
\label{suruseg}
\end{lemma}

\begin{proof} (i) Since $q(M)=1-\vert\pi_1(M)\vert\big(\frac{2}{3}
\big)^{2\:{\rm rk}(\pi_2(M))}$ by Lemma \ref{tlemma} and $0\leqq q(M)<1$
by Lemma \ref{terbelilemma}, taking $i=1,2,\dots$ and
$j=1,2,\dots$ we examine rational numbers of the form
$0<1-i\big(\frac{2}{3}\big)^{2j}<1$. For every $m>0$ there
exist lot of finitely presented groups $G$ such that
$\vert G\vert=m$ as well as $M_m\in\sm_{fin}$ so that $\pi_1(M_m)\cong G$
hence $\vert\pi_1(M_m)\vert=m$. Put $n_m:={\rm rk}(\pi_2(M_m))$ thus
$\big(\vert\pi_1(M_m)\vert,{\rm rk}(\pi_2(M_m))\big)=(m,n_m)$. Passing to a
subsequence if necessary, {\it suppose} that there exists
$M_1,M_2,\dots\in\sm_{fin}$ having the property
$q(M_m)=1-m\big(\frac{2}{3}\big)^{2n_m}\rightarrow 0$ as
$m\rightarrow+\infty$. Observe that this implies $n_m\rightarrow+\infty$ in
an approriate way too; more precisely writing
$m\big(\frac{2}{3}\big)^{2n_m}={\rm e}^{\log
m-n_m\log\frac{9}{4}}$ we see that $n_m\sim\frac{\log m}{\log\frac{9}{4}}$ as
$m\rightarrow+\infty$. Given a fixed inner point $r\in(0,1)$ by
assumption there exists $M_m\in\sm_{fin}$ so that $0<q(M_m)\leqq r$. For a
fixed $m$ the existence $M_m\in\sm_{fin}$ implies
that its smooth blow-up $M'_m=M_m\#\overline{\C P^2}$ exists
in $\sm_{fin}$ too, and $q(M_m)<q(M'_m)$ via Lemmata \ref{gammaformulak} and
\ref{tlemma}. Thus taking blow-ups of $M_m$ in $0<k<+\infty$ distinct points,
we can arrange $q\big(M^{(k-1)}_m\big)\leqq r\leqq q\big(M^{(k)}_m\big)$.
To close this scissor about $r$, observe that in fact $q(M_m^{(\ell)})=
1-m\big(\frac{2}{3}\big)^{2(n_m+\ell)}$ via the gluing formula in Lemma
\ref{gammaformulak} hence an $m$-independent estimate
$q\big(M^{(k)}_m\big)-q\big(M^{(k-1)}_m\big)=
m\big(\frac{4}{9}\big)^{n_m+k-1}\big(1-\frac{4}{9}\big)
\leqq\frac{5}{9}(\frac{4}{9}\big)^{k-1}$ holds;
thus at the expense of increasing $m$ and then $k$ we can also arrange
that $0<q\big(M^{(k)}_m\big)-q\big(M^{(k-1)}_m\big)<\varepsilon$. Consequently
$r\in(0,1)$ can be approximated as pleased
demonstrating that ${\rm im}\:q$ is indeed dense in $[0,1)$ hence the result.

(ii) This assertion is obvious (and has been stated only because it is dual
to the first assertion in an appropriate sense).
\end{proof}

\noindent If the hypothetical sequence of $4$-manifolds in Lemma
\ref{suruseg} exists then ${\rm im}\:q\subseteqq[0,1)$ is dense
consequently $\vert q^{-1}(0)\vert=+\infty$ by Lemma \ref{poincare2}
hence the smooth Poincar\'e conjecture fails by Lemma \ref{poincare1}.
Likewise $\vert q^{-1}\big(\frac{5}{9}\big)\vert=+\infty$ and one can
easily check via Lemma \ref{tlemma} and Freedman's classification
\cite{fre} that topologically $q^{-1}\big(\frac{5}{9}\big)= \big\{\C
P^2,\overline{\C P^2}\big\}$ over $\sm_{fin}$ hence at least one of these
spaces carries infinitely many smooth structures too.
\vspace{0.1in}

\noindent{\it Proof of Theorem \ref{fotetel3}}. Putting together
Lemmata \ref{egyszeresen-osszefuggo}, \ref{tlemma}, \ref{poincare1},
\ref{poincare2} and \ref{suruseg} the theorem follows.\hspace{1cm}$\square$

%%%%%%%%%%%%%%%%%%%%%%%%%%%%%%%%%%%%%%%%%%%%%%%%%

\section{The Einstein equation and comparison of representations}
\label{four}

%%%%%%%%%%%%%%%%%%%%%%%%%%%%%%%%%%%%%%%%%%%%%%%%%%%%%%

In the previous section making use of the smooth structure and orientation 
of a connected compact $4$-manifold alone, we have constructed a subfactor 
$0\not=\fii(M)\subsetneqq\fr$ encoding some information about topology and 
perhaps smoothness. In this section we turn the coin from the mathematical 
to its physical side and using additional geometric data we shall improve 
the subfactor to a normal subalgebra 
$0\not=\fii(M)\subset\fii(M,g)\subseteqq\fr$, which precisely in $4$ 
dimensions is related with an operator algebraic 
characterization of the Riemannian vacuum Einstein 
equation ${\rm Ric}=\Lambda g$.

{\it Construction of a canonical metric}. 
As a first step, let us observe that the plain manifold embedding 
$M\subset\fr$ in (\ref{beagyazas}) naturally enhances to a Riemannian 
embedding $(M,g)\subset(\fr, \re\tau)$. Consider the setup in Theorem 
\ref{fotetel1} again {\it i.e.}, take a connected compact oriented smooth 
$4$-manifold $M$ and consider its embedding into $\fr$ via (\ref{beagyazas}) 
mapping $x\in M$ to $P_x\in\fr$. Fix a point $x\in M$ and a tangent vector 
$X\in T_xM$. Take a $1$-parameter subgroup 
$\{\Phi_t\}_{t\in(-\varepsilon,+\varepsilon)}$ of diffeomorphisms such that 
the curve $t\mapsto\Phi_t(x)$ in $M$ is smooth and satisfies 
$\Phi_0(x)=x\in M$ and $\dot{\Phi}_0(x)=X\in T_xM$. 
Also consider the image curve $t\mapsto P(t):=i_M\Phi_t(x)$ in $\fr$. 
Hence $P(0)=P_x=i_Mx=i_M\Phi_0(x)\in\fr$ thus formally 
$\dot{P}(0)=\frac{\dd}{\dd t}i_M\Phi_t(x)\vert_{t=0}=
i_{M*}\dot{\Phi}_0(x)=i_{M*}X\in T_{P_x}\fr$. Moreover recall that 
$P(t)=\Phi^*_{-t}P_x\Phi^*_t$ hence also formally $\dot{P}(0)=[P_x,L_X]$, 
where $L_X\in\ed(\Omega^2(M;\C))$ is the Lie derivative constructed from
the infinitesimal generator, also denoted as $X\in C^\infty(M;TM)$, of
$\{\Phi_t\}_{t\in(-\varepsilon,+\varepsilon)}$. If 
$\{\Psi_t\}_{t\in(-\varepsilon,+\varepsilon)}$ is a similar 
family with corresponding projector curve $Q(t)=i_M\Psi_t(x)$ then we can
pick an intertwining diffeomorphism $\Lambda$ having the property 
$\Lambda(\Phi_t(x))=\Psi_t(x)$ for every $t$ 
hence $Q(t)=(\Lambda^{-1})^*P(t)\Lambda^*$ however such that
$(\Lambda^{-1})^*P(0)\Lambda^*=P(0)$ and $(\Lambda^{-1})^*\dot{P}(0)
\Lambda^*=\dot{P}(0)$. Consequently we compute that 
$\dot{Q}(0)=\frac{\dd}{\dd t}((\Lambda^{-1})^*P(t)
\Lambda^*)\vert_{t=0}=(\Lambda^{-1})^*\dot{P}(0)\Lambda^*=\dot{P}(0)$. 
Therefore the formal derivative of the embedding (\ref{beagyazas}) defined by 
$i_{M*}X:=[P_x,L_X]$ for every $x\in M$ and $X\in T_xM$, if exists, is 
well-defined {\it i.e.}, is independent of how $X\in T_xM$ has been extended 
to an $X\in C^\infty(M;TM)$. We know that always $P_x\not=0$ and 
assume that $X\not=0$ hence $L_X\not=0$ however
$[P_x,L_X]=0$ {\it i.e.}, the formal commutator is degenerate in this sense. 
This would imply that the image of $P_x$ in $\Omega^2(M;\C)$
is invariant under $L_X$; however this is not possible, for if 
$\varphi\in\Omega^2(M,x;\C)$ is any $2$-form vanishing at $x\in M$
then in general $L_X\varphi$ does not vanish there hence is not in
$\Omega^2(M,x;\C)$. Last but not least concerning existence, we know already 
that the $\C$-linear $\divideontimes$-automorphisms of $\fr$ in question 
{\it i.e.}, the $1$-parameter unitary group $t\mapsto{\rm Ad}_{\Phi_{-t}^*}$ is 
strongly continuous, hence by Stone's theorem its infinitesimal generator 
$L_X$, operating so far on $\Omega^2(M;\C)$ only, 
gives rise to a densly defined symmetric operator on $\ch$ too. Being a 
projection, it follows at once that $P_x\in\fr$ preserves the domain of $L_X$ 
hence the commutator $[P_x,L_X]$, as a densly defined antisymmetric operator on 
$\ch$, is meaningful. All of these allow us to use the non-degenerate 
scalar product $(A,B)=\tau(AB^\divideontimes)$ on $\fr$ 
(completing it to $\ch$ as above) to obtain a definite non-degenerate 
symmetric tensor $g$ on $M$ via pullback that is, 
put $g(X,Y):={\rm Re}(i_{M*}X,i_{M*}Y)$ yielding
\[g_x(X,Y):={\rm Re}\:\tau\big([P_x,L_X][P_x,L_Y]^\divideontimes\big)\]
for $X,Y\in T_xM$ (extended somehow to $X,Y\in C^\infty(M;TM)$ as above) 
and $x\in M$. This symmetric positive definite tensor field $g$ along $M$ is 
therefore well-defined, non-degenerate and exists  
hence gives rise to a Riemannian $4$-manifold $(M,g)$. Finally we remark that 
(\ref{beagyazas}) in its improved version 
$(M,g)\subset(\fr,\re\tau)$ is analogous to embedding Riemannian manifolds 
into Hilbert spaces via heat kernel techniques \cite{ber-bes-gal}.

{\it The Einstein condition}. Consider the 
complexified Hodge operator $*\in C^\infty(M;{\rm End}
(\wedge^2T^*M\otimes_\R\C))$ acting on $2$-forms along $(M,g)$. It readily 
follows from the left hand side inclusion in (\ref{izomorfizmus}) together 
with $C(M)\subset\fr$ that $*\in\fr$. This simple algebraic operator 
generates a dynamics on $\fr$ as follows.

\begin{definition} Let $M$ be a connected compact oriented smooth $4$-manifold 
with its induced Riemannian metric $g$ and corresponding Hodge operator 
$*\in\fr$ as above. That $*\not=1$ is self-adjoint and 
satisfies $*^2=1$ implies that there exists a basis in $\fr$ in which
$*=\left(\begin{smallmatrix}
                    1&0\\
                    0&-1
         \end{smallmatrix}\right)$. Therefore introducing the skew-Hermitian
operator $\log *:=\left(\begin{smallmatrix}
                    0&0\\
                    0&\sqrt{-1}\pi
             \end{smallmatrix}\right)$ for every $t\in\R$ we can set
$*^t:={\rm e}^{t\log*}$ which is a unitary in $\fr$. The corresponding
$1$-parameter family of $\C$-linear inner $\divideontimes$-automorphisms on
$\fr$ given by
\[A\longmapsto *^tA*^{-t}\]
for all $A\in\fr$ and $t\in\R$ introduces a non-trivial periodic dynamics
on $\fr$ what we call the {\rm Hodge dynamics}. Accordingly
$\big(\fr, \{{\rm Ad}_{*^t}\}_{t\in\R}\big)$ is a {\rm Hodge dynamical system 
on the hyperfinite ${\rm II}_1$ factor}.
\label{hodge}
\end{definition}

\noindent This naturally appearing dynamical system on $\fr$ 
can be used to characterize not only the plain 
embedding $M\subset\fr$ but even the geometric properties of the 
Riemannian embedding $(M,g)\subset(\fr,\re\tau)$ using the terminology 
of dynamical systems. The Hodge star as an element $*\in\fr$ is 
self-adjoint, satisfies $1\not=*$ but $1=*^2$ hence generates 
an Abelian von Neumann subalgebra $\langle*\rangle\subset\fr$ isomorphic 
to $\C1\oplus\C*$; consequently its relative commutant 
$\fii(M,g):=\langle*\rangle'\cap\fr$ extends $\langle*\rangle$ to a 
by construction normal subalgebra of $\fr$ ({\it i.e.}, is equal to its double 
relative commutant). It can be identified with the {\it fixed-point-subalgebra} 
of the Hodge dynamics: $\fii(M,g)=\big\{A\in\fr\:\vert\: {\rm Ad}_{*^t}A=A\:
\mbox{for every $t\in\R$}\big\}$. Conversely, every normal subalgebra of $\fr$ 
arises as the fixed-point-subalgebra of a periodic inner 
$\C$-linear $\divideontimes$-automorphism of $\fr$, cf. 
\cite[Theorem 3.1]{sto2}.

\begin{lemma}
Let $M$ be a connected compact oriented smooth $4$-manifold and consider its
embedding (\ref{beagyazas}) into $\fr$. Take 
the induced Riemannian metric $g$ on $M$ and Hodge dynamics 
$\{{\rm Ad}_{*^t}\}_{t\in\R}$ on $\fr$ as well as the subfactor 
$\fii(M)\subset\fr$ as in (\ref{idef}) and the Riemannian curvature 
tensor (\ref{gorbulet}) satisfying $R_g\in\fr$. 

The projections of the embedding satisfy $P_x\in\fii(M,g)$ for 
every $x\in M$ {\it i.e.}, $M\subset\fr$ is pointwisely preserved by the Hodge 
dynamics; moreover $\fii(M)\subset\fii(M,g)$ 
{\it i.e.}, the subfactor generated by $M$ is also pointwisely fixed 
by the Hodge dynamics. Finally $R_g\in\fii(M,g)$ if and only 
if $(M,g)$ is Einstein.
\label{einstein}
\end{lemma}

\begin{proof} The complexified Hodge star is a gauge transformation of 
$\wedge^2T^*M\otimes_\R\C$ hence we know from the general theory summarized 
just above (\ref{automorfizmus}) that it leaves $M$ pointwise fixed. 
Consequently the corresponding $\C$-linear $\divideontimes$-automorphism of 
$\fr$ satisfies $*P_x*=P_x$ for the elementary projections $P_x\in\fr$ 
assigned to a point $x\in M$ via (\ref{beagyazas}) that is, $P_x\in\fii(M,g)$. 
Moreover pick $A\in\fii(M)$ and using the decomposition 
(\ref{idef}) write it in the form 
$A=B^\divideontimes C$ where $B,C\in I(M)$. Then by (\ref{ideal}) this 
condition is equivalent to $BP_x=0$ for every $x\in M$ consequently 
$0=*^t(BP_x)*^{-t}=(*^tB*^{-t})(*^tP_x*^{-t})=(*^tB*^{-t})P_x$ for every 
$x\in M$, thus $*^tB*^{-t}\in I(M)$ too and likewise for $C\in I(M)$. 
Therefore $*^t(\fii (M))*^{-t}=\fii(M)$ that is, $\fii(M)$ is invariant under 
the Hodge dynamics. However we can say more. Combining 
$*P_x*=P_x$ with $*^2=1$ we get $*P_x=P_x*$ hence the two operators commute. 
In this context $B\in I(M)$ {\it i.e.}, $BP_x=0$ implying 
$P_xB^\divideontimes=0$, 
can be interpreted as $B^\divideontimes$ is an eigenvector of $P_x$ with 
zero eigenvalue; hence $B^\divideontimes$ is also an eigenvector for $*$ with 
possible eigenvalues $\pm 1$. The same holds for 
$C^\divideontimes$ moreover $B^\divideontimes+C^\divideontimes$ too, hence 
both $B^\divideontimes$ and $C^\divideontimes$ are eigenvectors with the same 
eigenvalue. Thus taking $\fii(M)\ni A=B^\divideontimes C\in I(M)^\divideontimes 
I(M)$ we compute $*^tA*^{-t}=*^t(B^\divideontimes C)*^{-t}=
(*^tB^\divideontimes)(*^tC^\divideontimes)^\divideontimes=
(\pm 1)^tB^\divideontimes((\pm 1)^tC^\divideontimes)^\divideontimes=
B^\divideontimes C=A$. We conclude that $\fii(M)$ is in fact pointwisely 
fixed by the Hodge dynamics hence $\fii(M)\subset\fii(M,g)$ too. 
Finally the assertion on the Einstein condition is straightforward, as 
following \cite{sin-tho} we notice at once by comparing the curvature $R_g$ in 
(\ref{gorbulet}) and 
\[*=\left(\begin{matrix}
                1 & 0\\
                0 & -1
          \end{matrix}\right)
\:\:\:\:\::\:\:\:\:\:\begin{matrix}\Omega^+(M;\C)\\
                                                  \bigoplus\\
                                                 \Omega^-(M;\C)
                                            \end{matrix}
         \:\:\:\longrightarrow\:\:\:\begin{matrix}\Omega^+(M;\C)\\
                                                  \bigoplus\\
                                                 \Omega^-(M;\C)
                                            \end{matrix}\]
that $g$ is Einstein {\it i.e.}, the traceless Ricci part of $R_g$ in 
(\ref{gorbulet}) vanishes if and only if $*$ commutes with $R_g$ 
which together with $*^2=1$ implies that $*R_g*=R_g$ which means that 
$R_g\in\fii(M,g)$. 
\end{proof}

\noindent{\it Proof of Theorem \ref{fotetel4}}. Taking into account the 
construction of the metric $g$ above, Definition \ref{hodge} and Lemma 
\ref{einstein} the theorem follows.\hspace{11.5cm}$\square$ 
\vspace{0.01in}

\noindent {\it Comparison of representations and physics}. Lemma 
\ref{einstein} can be regarded as a sort of compatibility result between 
two pieces of data on the hyperfinite ${\rm II}_1$ factor namely a 
subfactor $0\not=\fii(M)\subsetneqq\fr$ for $M$ and a normal subalgebra 
$0\not=\fii(M,g)\subseteqq\fr$ for $(M,g)$ in the sense that 
$\fii(M)\subset\fii(M,g)$ and additionally 
$R_g\in\fii(M,g)$ if and only if $(M,g)$ is 
Einstein. A different, somewhat more physical way of 
understanding compatibility between these two structures is as follows. 
A connected, compact, oriented, smooth $4$-manifold $M$ also gives rise 
to a representation $\rho_M$ of $\fr$ as in Lemma \ref{terbelilemma} and 
a Hodge dynamical system $\big(\fr, \{{\rm Ad}_{*^t}\}_{t\in\R}\big)$ 
introduced in Definition \ref{hodge}. The mathematical fact that $\fr$ 
admits many {\it inequivalent} representations ({\it i.e.}, the failure of the 
Stone--von Neumann representation theorem in this case) can be 
interpreted in the framework of {\it algebraic quantum field theory} 
\cite{bra-rob,haa} as saying that $\fr$ is an operator algebra of a quantum 
system possessing infinitely many degrees of freedom like a quantum 
statistical ensemble. In this context, as well as 
recalling \cite{con-rov}, the temptation here to interpret $\fr$ as the 
operator algebra of a relativistic quantum field theory at non-zero 
temperature involving gravity, is supported by the following further 
observations (also cf. \cite{cha-lon-pen-wit}). On the one 
hand $\fr$ contains curvature tensors, the key objects of general 
relativity. On the other hand the periodicity of the Hodge dynamics on 
this operator algebra {\it i.e.}, the plain mathematical property $*^2=1$ of 
the Hodge star on $2$-forms in $4$ dimensions, can also be interpreted 
along these lines in the well-known way, as the presence of a 
temperature in a statistical ensemble (cf. {\it e.g.} \cite{itz-zub}). This 
temperature is the inverse of the period hence is equal to 
$\frac{1}{2}T_{\rm Planck}$ in natural units. This permits one to analyze 
the interference between the aforementioned structures 
within the realm of the theory of thermal equilibrium states in 
algebraic quantum field theory. For stationarity and 
stability are expected properties of thermal equilibrium states 
(see \cite[Section 5.4]{bra-rob}, \cite[Section V.3]{haa}), as a first 
step in this analysis we record here the following stationarity and stability 
property of the representations, more precisely their corresponding states on 
$\fr$, against their induced Hodge dynamics and their perturbations on $\fr$.

\begin{lemma} Consider the embedding $M\subset\fr$ by (\ref{beagyazas}) 
more precisely the induced oriented Riemannian 
$4$-manifold $(M,g)\subset(\fr,{\rm Re}\tau)$ as above. Also consider the 
associated Hodge dynamical system $\big(\fr,\{{\rm Ad}_{*^t}\}_{t\in\R}\big)$ 
as in Definition \ref{hodge}.

Then the state $F_{M,g}:\fr\rightarrow\C$ in the proof of Lemma 
\ref{terbelilemma} provided by $(M,g)$ is stationary under the Hodge dynamics 
{\it i.e.}, $F_{M,g}(*^tA*^{-t})= F_{M,g}(A)$ for every $A\in\fr$ and $t\in\R$.

Moreover take a unitary element $*'\in\fr$ and let
$\big(\fr,\{{\rm Ad}_{(*')^t}\}_{t\in\R}\big)$ be a ``nearby''
dynamical system in the sense that it preserves $M\subset\fr$ and satisfies
$(*')^p=1$ {\it i.e.}, is periodic with $1\leqq p<+\infty$. Then there exists
a corresponding ``nearby'' state $F_{M,g'}$ on $\fr$, yet inducing the same
representation $\rho_M$ of $\fr$, which is stationary under the ``nearby''
dynamics.
\label{allapot}
\end{lemma}

\begin{proof} Recall that $F_{M,g}:\fr\rightarrow\C$ has been defined in
Lemma \ref{terbelilemma} as $F_{M,g}(A)=\int_Mi_M^*\tau(AP_x)\mu_g$. Taking
into account that $*$ commutes with $P_x$ as in Lemma \ref{einstein} 
and using the cyclic property of the trace it follows at once that 
$F_{M,g}(*^tA*^{-t})=F_{M,g}(A)$ for every $A\in\fr$ and $t\in\R$.

The ``perturbed'' dynamics generated by $*'$ as a $1$-parameter inner 
$\divideontimes$-automorphisms of $\fr$ preserves $M$ by assumption 
hence admits a unique decomposition into a $1$-parameter family of gauge 
transformations and diffeomorphisms according to (\ref{automorfizmus}) hence 
${\rm Ad}_{(*')^t}={\rm Ad}_{\gamma_{\:t}}{\rm Ad}_{\Phi_{-t}^*}$. Being the
''perturbed'' dynamics periodic along $M$ its orbits are compact
consequently there exists a $*'$-averaged metric $g'$ on $M$ whose volume
form $\mu_{g'}$ is preserved by the perturbed dynamics. Thus 
\begin{eqnarray}
F_{M,g'}\big((*')^tA(*')^{-t}\big)&=&
\int\limits_{x\in M}i_M^*\tau\big(({\rm Ad}_{(*')^t}A)P_x\big)\mu_{g'}(x)=
\int\limits_{x\in M}i_M^*\tau\big(A({\rm Ad}_{(*')^{-t}}P_x)\big)\mu_{g'}(x)
\nonumber\\
&=&\int\limits_{x\in M}i_M^*\tau\big(A({\rm Ad}_{\gamma_{-t}}
{\rm Ad}_{\Phi_t^*}P_x)\big)\mu_{g'}(x)=
\int\limits_{x\in M}i_M^*\tau\big(AP_{\Phi_{-t}(x)}\big)\mu_{g'}(x)\nonumber\\
&=&\int\limits_{\Phi_{-t}(x)\in M}i^*_M\tau\big(AP_{\Phi_{-t}(x)}\big)
\mu_{g'}(\Phi_{-t}(x))\nonumber\\
&=&F_{M,g'}(A)\nonumber
\end{eqnarray}
demonstrating that $F_{M,g'}$ is stationary under the ''perturbed'' dynamics. 
Finally we have observed already in the proof of Lemma \ref{terbelilemma} 
that the ideal $0\not=I(M)\subsetneqq\fr$ in (\ref{ideal}) 
consisting of elements satisfying $F_{M,g'}(A^\divideontimes A)=0$ is 
independent of the metric $g'$ hence the representation of $\fr$ induced by 
$F_{M,g'}$ coincides with that one induced by $F_{M,g}$ {\it i.e.}, with the 
representation $\rho_M$ of $\fr$ constructed in 
Lemma \ref{terbelilemma} hence the result.
\end{proof}

\noindent Following \cite[Chapter 8]{ana-pop} appropriately finite 
representations of the hyperfinite ${\rm II}_1$ factor are classified by 
their Murray--von Neumann coupling constants or $\fr$-dimensions, taking 
all possible values in the real half-line $[0,+\infty)$, see 
\cite[Proposition 8.6.1]{ana-pop}. A representation, uniquely 
characterized by its $\fr$-dimension $y\in[0,+\infty)$ as its numerical 
invariant, naturally decomposes according to $y=[y]+\{y\}$ {\it i.e.}, splits 
into its integer part with $\fr$-dimension 
$[y]\in\{0,1,2,\dots\}\subset[0,+\infty)$ containing copies of the 
representation having $\fr$-dimension precisely $1$ and into its 
fractional part given by $\{y\}\in(0,1)\subset[0,+\infty)$ describing 
another representation whose $\fr$-dimension falls within the open unit 
interval. Consequently it is enough to understand those representations 
which belong to the closed unit interval $[0,1]\subset[0,+\infty)$ only. 
Of course the representation characterized by $0\in[0,1]$ is just the 
{\it trivial representation}. The representation having $\fr$-dimension 
precisely $1\in[0,1]$ is the {\it standard representation} $\pi$ of 
$\fr$ on itself by (left-)multiplications (as above). This is the 
best-known non-trivial representation possessing the following 
remarkable properties:

\begin{itemize}

\item[(i)] within the framework of the {\it Gelfand--Naimark--Segal} 
(GNS) construction, the unique standard representation $\pi$ of $\fr$ on 
a Hilbert space $\ch$ can be obtained from a distinguished faithful 
state, namely the unique finite trace $\tau$ on $\fr$ (see \cite[Section 
2.6]{ana-pop});

\item[(ii)] the {\it Tomita--Takesaki} modular theory is applicable to $\pi$
and the corresponding modular operator $\Delta$ renders $\fr$
a dynamical system, however this modular dynamics is trivial because
$\tau$ is tracial (see \cite[p. 90, Subsection 2.5.2]{bra-rob});

\item[(iii)] the state $\tau$ is a {\it Kubo--Martin--Schwinger} (KMS) 
state on $\fr$ with respect to the modular dynamics, however in a 
trivial way and the formal KMS temperature of this state is infinite, 
both because $\tau$ is tracial (cf. \cite[Section 5.3]{bra-rob}, 
\cite[Section V.2]{haa}).

\end{itemize}

\noindent In Section \ref{two} using smooth $4$-manifolds $M$ 
we have constructed an immense class of geometric representations 
$\rho_M$ whose $\fr$-dimensions $q(M)$ fall into $[0,1)$. What about 
properties (i)-(iii) concerning these fractional representations? 
Interestingly, we can can exhibit a list of analogous properties:

\begin{itemize}

\item[(iv)] within the GNS construction every connected compact
oriented smooth $4$-manifold $M$ gives rise to a
representation $\rho_M$ of $\fr$ on a Hilbert space $\ci(M)^\perp$ obtained
from a non-faithful state $F_{M,g}$ on $\fr$ (cf. Lemma \ref{terbelilemma});

\item[(v)] to every $\rho_M$ as above there exists a unitary operator
$*\in\fr$ which renders $\fr$ a dynamical system such that this 
Hodge dynamics is already non-trivial nevertheless always satisfies $*^2=1$ 
(cf. Definition \ref{hodge});

\item[(vi)] the state $F_{M,g}$ is not faithful thus surely not 
KMS (cf. \cite[p. 207]{haa}) but is invariant under, and the 
corresponding representation $\rho_M$ is stable against small 
perturbations of, the Hodge dynamics, hence $F_{M,g}$ nevertheless describes a 
thermal equilibrium state with respect to this periodic dynamics at a uniform 
formal temperature $\frac{1}{2}T_{\rm Planck}$ (cf. Lemma \ref{allapot} 
and the discussion before it). 
\end{itemize}

\noindent This comprehensive view of representations strongly motivates 
the following physical picture: there exists a unique physical system 
whose operator algebra is $\fr$ but this system exhibits different 
physical phases corresponding to inequivalent representations of $\fr$. 
Therefore, as a working hypothesis, it is challenging to physically 
interpret the quite circular interaction between $\fr$ and $M$ unfolded 
here as follows: the unique abstract triple $(\fr,\ch,\pi)$ 
describes the {\it quantum phase} of, while a highly non-unique triple 
$(\fr,\ci(M)^\perp,\rho_M)$ corresponding to a $4$ dimensional Riemannian 
manifold $(M,g)$, if Einstein, describes a particular state from the usual {\it 
classical phase} of Riemannian vacuum general relativity precisely in 
$4$ dimensions. (A non-vacuum but usual choice for $(M,g)$ is the FLRW 
solution with or without cosmological constant.) We can display symbolically 
this passage as 
\vspace{0.1in} 
\[\begin{matrix} 
(\fr,\ch,\pi)&\Longrightarrow & \big(\fii(M,g)\subset\fr\:,\: 
\ci(M)^\perp\subset\ch\:,\: \rho_M=\pi\vert_{\ci(M)^\perp}\big)\\ 
\hspace{-1cm}\Updownarrow & &\Updownarrow\\ 
\mbox{The ``quantum 
space-time''} & &\mbox{A particular $4$ dimensional Riemannian vacuum 
space-time $(M,g)$}\\ 
\mbox{at infinite temperature} &&
\mbox{at temperature $\frac{1}{2}T_{\rm Planck}$} 
\end{matrix}\] 
\vspace{0.1in}

\noindent having in mind a sort of phase transition 
from the quantum to the classical phase of the theory via spontaneous 
symmetry breaking by cooling (or by a spontaneous jump from the 
unique Tomita--Takesaki to a particular Hodge dynamics on $\fr$). Note 
that this transition from the unique quantum regime $(\fr,\ch,\pi)$ to a 
particular $4$ dimensional classical vacuum regime $(M,g)$ given by 
another representation $(\fr,\ci(M)^\perp,\rho_M)$ has been captured in 
the framework of algebraic quantum field theory 
\cite[Subsection 4.3.4]{bra-rob}, \cite[Subsection V.1.5]{haa} as 
switching from the unique representation $\pi$ to a different particular 
representation $\rho_M$ of the same algebra $\fr$. One can also formally 
label the transition with $\frac{1}{2}T_{\rm Planck}\approx7.06\times 
10^{31}$ K which is the formal temperature associated with $\rho_M$; 
this high temperature is reasonable if we keep in mind that $\pi$ 
corresponds to infinite temperature. Finally observe that during this 
spontaneous symmetry breaking procedure the original gauge group 
${\rm U}(\ch)\cap\fr\subset{\rm Aut}\:\fr$ breaks down to its subgroup 
${\rm Diff}^+(M)$ justifying the terminology. 
\vspace{0.1in}

\noindent{\bf Declarations and statements}. All the not-referenced results
in this work are fully the author's own contribution. There are no conflict of
interest to declare that are relevant to the content of this article.
The work meets all ethical standards applicable here. No funds, grants,
or other financial supports were received. Data sharing is not
applicable to this article as no datasets were generated or analyzed
during the underlying study.

\end{document}